%% file: IEEE_TGCN-RP-A2-23-0120.R2_Final.tex
\documentclass[12pt,draftcls, onecolumn]{ieeetran}
\usepackage{graphics}
\usepackage{graphicx}
\usepackage{amssymb,amsmath}
\usepackage{amsmath}
\usepackage{amsthm}
\usepackage{url}
\usepackage[utf8]{inputenc}
\usepackage{algorithm}
\usepackage{hyperref}
\usepackage{indentfirst}
\usepackage{setspace}
\usepackage{cases}

\usepackage{booktabs}
\usepackage{multicol}
\usepackage{comment}
\usepackage{mathtools}
\usepackage[noend]{algpseudocode}

\usepackage{cite}

\newtheorem{lemma}{Lemma}

\theoremstyle{remark}

\usepackage{stfloats}
\usepackage{setspace}
\usepackage{lettrine}

\usepackage{suffix}
\usepackage{mathtools}

\DeclarePairedDelimiterX\MeijerM[3]{\lparen}{\rparen}%
{\begin{smallmatrix}#1 \\ #2\end{smallmatrix}\delimsize\vert\,#3}

\newcommand\MeijerG[8][]{%
  G^{\,#2,#3}_{#4,#5}\MeijerM[#1]{#6}{#7}{#8}}

\WithSuffix\newcommand\MeijerG*[7]{%
  G^{\,#1,#2}_{#3,#4}\MeijerM*{#5}{#6}{#7}}

\newcounter{problem}
\newcommand{\Pn}{\mathcal{P}_{\arabic{problem}}}
\stepcounter{problem}

\usepackage{booktabs}
\usepackage[nolist]{acronym}
\usepackage{lipsum}
\usepackage[acronym,,nohypertypes={acronym,notation}]{glossaries}
\input{acronyms}

\newcommand{\diag}{{\rm diag}} 
\newcommand{\D}{\mathbf{D}}
\newcommand{\G}{\mathbf{G}}
\newcommand{\F}{\mathbf{F}}
\newcommand{\etaee}{\eta_{{\rm EE}}}
\newcommand{\pbs}{P_{{\rm BS}}}
\newcommand{\pris}{P_{{\rm RIS}}}
\newcommand{\ptx}{P_{{\rm TX}}}

\newcommand{\customstar}{\mathbin{\text{\raisebox{-0.4ex}{\scalebox{0.9}{$\star$}}}}}

\begin{document}
\title{Energy Efficiency Maximization for Intelligent Surfaces Aided Massive MIMO with Zero Forcing}
\author{{Wilson de Souza Junior} and Taufik Abrão
\thanks{This work was supported by the National Council for Scientific and Technological Development (CNPq) of Brazil, Grant: 310681/2019-7, by CAPES, Grant: FC001, and by Londrina State University (UEL), Brazil.}
\thanks{W. Junior and T. Abrão are with the Department of  Electrical Engineering (DEEL). State University of Londrina  (UEL). Po.Box 10.011, CEP:86057-970, Londrina, PR, Brazil.  Email: wilsoonjr98@gmail.com; taufik@uel.br}
}
\maketitle

\begin{abstract}
In this work, we address the \gls{ee} maximization problem in a downlink communication system utilizing \gls{ris} in a multi-user \gls{m-mimo} setup with \gls{zf} precoding. The channel between the \gls{bs} and \gls{ris} operates under a Rician fading with Rician factor $K_1$. Since systematically optimizing the \gls{ris} phase shifts in each channel coherence time interval is challenging and burdensome, we employ the statistical \gls{csi}-based optimization strategy to alleviate this overhead. By treating the \gls{ris} phase shifts matrix as a constant over multiple channel coherence time intervals, we can reduce the computational complexity while maintaining an interesting performance. Based on an \gls{er} lower bound closed-form, the \gls{ee} optimization problem is formulated. Such a problem is non-convex and challenging to tackle due to the coupled variables. To circumvent such an obstacle, we explore the sequential optimization approach where the power allocation vector $\mathbf{p}$, the number of antennas $M$, and the \gls{ris} phase shifts $\mathbf{v}$ are {separated} and sequentially solved iteratively until convergence. With the help of the Lagrangian dual method, \gls{fp} techniques, {and supported by Lemma \ref{lemma:1},} insightful {compact} closed-form expressions {for each of the three} optimization variables are derived. Simulation results validate the effectiveness of the proposed method across different generalized channel scenarios, including \gls{nlos} ($K_1=0$) and partially \gls{los} ($K_1\neq0$) conditions. Our numerical results demonstrate an impressive performance of {the proposed} Statistical \gls{csi}-based {\gls{ee}} optimization {method}, achieving $\approx 92\%$ of the performance attained through {perfect} instantaneous \gls{csi}-based \gls{ee} optimization. This underscores its potential to significantly reduce power consumption, decrease the number of active antennas at the base station, and effectively incorporate \gls{ris} structure in \gls{m-mimo} communication setup with just statistical \gls{csi} knowledge.
\end{abstract} 
\acresetall
\smallskip
\noindent 


\section{Introduction}
\lettrine[findent=2pt]{\textbf{R}}{}econfigurable intelligent surfaces is one of the most promising recent techniques that have the potential to integrate future communications systems (beyond fifth generation, 6G, etc.) due to its potential to improve the transmission by creating a ``virtual" reconfigurable communication link. Specifically, the \gls{ris} is composed of several individually scattering elements which are artificial meta-material structures that can reflect incident electromagnetic waves and can be configured to increase the signal level in a specific direction\cite{9122596,8910627}.

On the other hand, to enhance the \gls{se} and to serve a plurality of \gls{ue} using the same physical layer resources, \gls{m-mimo} is essential; however, this technology resorts to hundreds of antennas operating at the \gls{bs} \cite{marzetta2016fundamentals}, which can significantly increase the consumption. 

\gls{ris} and \gls{m-mimo} technologies hold great promise for the future of communication systems. As these technologies are set to play a vital role in the next-generation wireless networks, which will need to support exceptionally high data rates and accommodate a vast number of connected \gls{ue}s \cite{9349624,9397776}, there is a growing concern about their impact on energy consumption. Hence, it is imperative to focus on \gls{ee}, measured in bits-per-joule, as a pivotal performance metric to ensure the development of environmentally friendly and sustainable communication systems.

{An extensive number of works such as in \cite{He2022, RangLiu2022,9814527,9963962,8982186,10136805,9246254,9681803,EEdebbah, EEzeng,EEli,9548940,fotock2023energy, Zhi2022,STkangda,10167745,10072840} have studied the \gls{ris} deployment impact in \gls{m-mimo} systems under different objectives, such as \gls{ris}-assisted sensing \cite{He2022, RangLiu2022,9814527}, minimizing the total transmit power \cite{9963962}, maximizing the \gls{wsr} \cite{8982186}, maximizing the minimum rate \cite{10136805,9246254,9681803}, and maximizing the \gls{ee} \cite{EEdebbah, EEzeng, EEli,9548940,fotock2023energy}. All of the above-cited studies clearly demonstrate the immense potential and versatility of \gls{ris} within wireless systems.}

It is worth noting that all the previously mentioned studies assume the transmitter to access the instantaneous \gls{csi}. However, in practical applications, obtaining precise instantaneous \gls{csi} in \gls{ris}-enhanced \gls{m-mimo} systems can be quite challenging. This challenge becomes even more pronounced in scenarios with high mobility and short channel coherence times. To address this issue, in \cite{Zhi2022}, the authors derived a rigorous tight closed-form expression for the \gls{er} of a multi-user \gls{ris}-aided \gls{m-mimo} system employing \gls{zf} precoding over Rician fading. This analysis takes into account the optimization of \gls{ris} phase shifts based on derived statistical closed-form expression, offering a more practical and robust approach in scenarios where instantaneous \gls{csi} may not be readily available or reliable.  In \cite{STkangda}, the authors exploit the long-term statistical \gls{csi} to optimize the \gls{ris} phase shifts aiming to maximize the achievable rate with \gls{mrc}, where a genetic algorithm is deployed.  { In \cite{10167745}, the authors leverage the statistical \gls{csi} optimization in order to optimize the beamforming at \gls{bs}, passive beamforming at \gls{ris}, and the power of the \gls{ue}s, while considering maximum \gls{ue}s's exposure to electromagnetic fields.  In \cite{10072840}, the authors investigated the \gls{ee} optimization in a \gls{ris}-assisted \gls{m-mimo} scenario with one single user based on statistical-\gls{csi} optimization. The authors show the potential of optimizing the \gls{ris} phase shifts, \gls{bs} transmit beamforming, and the \gls{ue}s's power. }

To emphasize and highlight the significance and advantages of the statistical \gls{csi}-based optimization approach, Figure \ref{fig:TS} demonstrates the substantial time overhead associated with adopting instantaneous \gls{csi}-based optimization compared to statistical \gls{csi}-based optimization methods. With the statistical \gls{csi}-based optimization, there is no need to continuously optimize the \gls{ris} phase shifts during various subsequent channel coherence times. This reduction in overhead leads to enhanced payload data efficiency, making it a practical choice for challenging scenarios in \gls{ris}-assisted \gls{m-mimo} systems. Hence, the long-term statistical \gls{csi}-based optimization strategy could be more practical and feasible \cite{STkangda}.

{Furthermore, joint beamforming and phase shift optimization are extremely significant in order to improve the performance of wireless communication systems, particularly in \gls{ris}-aided \gls{m-mimo} systems. On the other hand, several works in recent literature deploy a \gls{zf}-based precoding scheme instead of directly optimizing the precoding policy.  Herein, we adopt \gls{zf} linear precoding motivated by two facts. First, we have opted for the well-established linear precoder \gls{zf} to circumvent the challenges posed by optimal precoding design in \gls{ris}-aided \gls{m-mimo} design. Second, our work is primarily motivated by the concept of statistical \gls{csi} optimization, which has the potential to significantly alleviate the overhead associated with the \gls{ris} phase shifts optimization procedure within each channel coherence time. We are interested in revealing the improvements attainable through instantaneous \gls{csi} optimization and emphasize the relevance of the statistical \gls{csi} optimization approach in \gls{ris}-aided \gls{m-mimo} design since such an approach can substantially reduce system complexity with only a marginal performance degradation.}

\vspace{2mm}
\noindent{\textbf{\textit{Novelty}}}: {To the best of our knowledge, the \gls{ee} optimization problem with statistical \gls{csi} under Rician fading for multi-user \gls{ris}-aided \gls{m-mimo} systems with \gls{zf} has not yet been thoroughly investigated. This knowledge gap has motivated us to embark on resolving this important challenge. In our work, unlike previous works, such as  \cite{EEdebbah, EEzeng, 9548940, fotock2023energy}, the \gls{ee} optimization problem is approached in a multi-user \gls{ris}-aided \gls{m-mimo} with statistical \gls{csi} where an alternating optimization strategy is proposed to find the number of \gls{bs} antennas ($M$), the \gls{ris} phase shifts ($\bf v$), and the transmit power allocation ($\bf p$) to maximize the \gls{ee}.
Moreover, our approach distinguishes itself from that presented in \cite{10167745} by not only considering statistical \gls{csi} but also incorporating \gls{zf} precoding and the optimization of active antennas at the \gls{bs}. Similarly, in \cite{10072840}, although the consideration of statistical \gls{csi} is similar, the crucial multi-user scenario was overlooked. Our research goes a step further by providing novel closed-form solutions for all optimization variables, thereby expanding the scope of active antennas and power allocation beyond the established boundaries of the traditional \gls{m-mimo} scenario for the domain of \gls{ris}-assisted configurations. 
Lastly, {but equally important}, in contrast to the conventional {\gls{ee}} solutions found in existing literature, which often rely on {a couple of optimization techniques, namely} \gls{sdr},  \gls{sfp}, or Manifold optimization, our study, {supported by Lemma \ref{lemma:1},} introduces a novel and straightforward closed-form solution 
for the quadratic programming problem involving the non-convex unit modulus constraint.}

\vspace{2mm}
\begin{figure}[!ht]
\vspace{-0.2cm}
\normalsize
\begin{center}
\includegraphics[width=1.105\columnwidth]{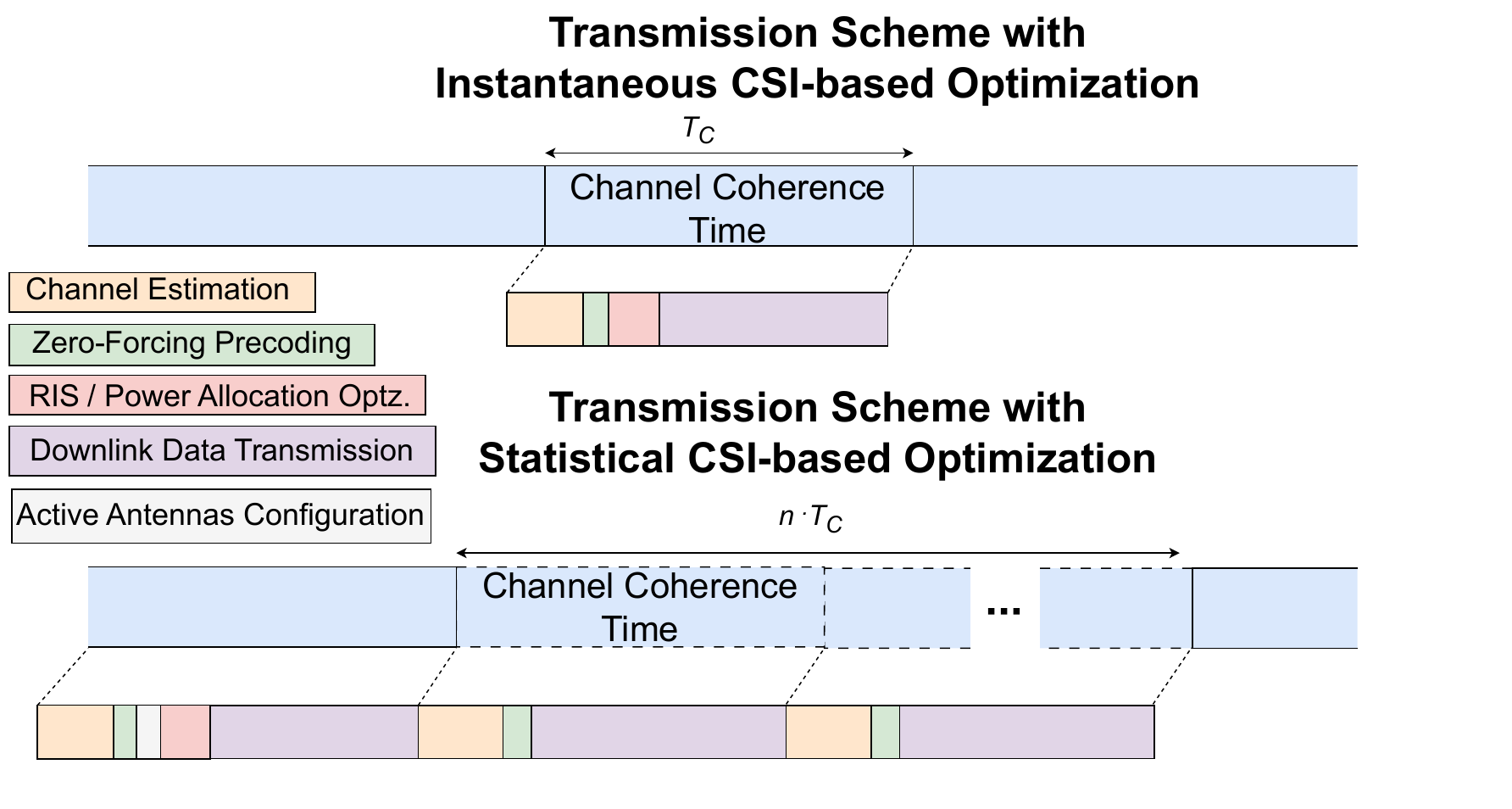}
\end{center}
\caption{ Illustration of the 
advantages by adopting Statistical \gls{csi}-based parameters optimization, where the \gls{ris} phase shifts vector ($\bf v$), transmit power allocation ($\bf p$), and the number of
BS antennas ($M$) remain fixed for the \gls{ue}s over multiple channel coherence times ($nT_c$); differently, for the Instantaneous \gls{csi}-based parameter optimization approach, where the channel estimation, precoding, \gls{ris} phase shifts, and transmit power must be updated at each $T_c$.}
\label{fig:TS}
\end{figure}

\vspace{2mm}
The main \textbf{\textit{contributions}} of this work are fourfold:

\begin{itemize} 
\item In this study, we have formulated a comprehensive non-convex optimization problem encompassing multiple key variables {in a multi-user \gls{ris}-aided \gls{m-mimo} system}, such as the \gls{ue}s's transmit power, the quantity of active base station antennas, and the phase shifts associated with the \gls{ris} elements. We also have addressed the constraints linked to the overall power budget, the assurance of \gls{qos}, and the necessity for unit module values since we consider only passive \gls{ris}. In addition, we have embraced the utilization of linear \gls{zf} precoding, wherein the feasibility of \gls{zf} is a pivotal consideration. Our approach extends further
by proposing an innovative solution that culminates in simple yet accurately derived closed-form expressions for all the variables at hand. {Specifically, for the phase shifts on the \gls{ris} optimization, the closed-form solution is an interesting alternative for the \gls{sdr}, \gls{sfp}, and Manifold optimization.}

\item Different from the previous works \cite{EEdebbah,EEzeng,EEli,9548940,fotock2023energy,10167745,10072840},  the formulated \glspl{ee} maximization problem relies upon only the statistical \gls{csi} estimates, which can be beneficial in short channel coherence time wireless scenarios; indeed, such an approach does not need precise instantaneous \gls{csi} estimates given optimizing the \gls{ris} phase shifts and power allocation, reducing the overhead imposed by the optimization procedure of these variables in every channel coherence time {and easing the practical implementation of the \gls{ris}-assisted system}. 
    
\item To deal with this problem, we apply the Lagrangian dual method with {\gls{fp} techniques to solve the non-convex problem}. {Specifically for the \gls{ris} phase shifts problem, we proposed an innovative \gls{fp}-based \gls{ris} phase shifts optimization, which enables us to find a closed-form expression for each element of \gls{ris}. To support our findings, the complexity of the proposed} antenna, power allocation, and \gls{ris} phase shifts optimization algorithm has been accurately investigated.

\item Extensive numerical results corroborate the effectiveness and efficiency of the proposed analytical \gls{ee} optimization method for \gls{ris}-aided \gls{m-mimo} systems with \gls{zf} operating under generalized Rician channels. {The proposed method has exhibited superior system performance in comparison to the gradient descent method. Moreover, its potential has been highlighted through promising results in contrast to the optimization based on instantaneous \gls{csi}-based optimization strategy.}
\end{itemize}

The remainder of this paper is organized as follows. Section \ref{sec:sys_model} describes the \gls{ris}-aided \gls{m-mimo} system model. The optimization problem and the overall proposed analytical optimization methods and algorithms are 
developed and discussed in Section \ref{sec:EEoptzproblem}. We present numerical results corroborating our findings in Section \ref{sec:simul}. Final remarks are offered in Section \ref{sec:concl}.

\textit{Notations:} Scalars, vectors, and matrices are denoted by the lower-case, bold-face lower-case, and bold-face upper-case letters, respectively; $\mathbb{C}^{A\times B}$ denotes the space of $A$ rows and $B$ columns complex matrices; $|\cdot|$ denote the absolute operator; $\diag(\cdot)$ denotes the diagonal operator; ${\rm Tr}(\cdot)$ is trace operator; $\left[x\right]^+$ denotes $\max(x,0)$; $[x]^a_b$ denotes $\min(\max(x,b),a)$; $[\cdot]^*$ denote the conjugate operator; $\angle[\cdot]$ denotes the phase of the complex argument; $j \triangleq \sqrt{-1}$; $\mathbb{E}\left[\cdot\right]$ is the expectation operator; $[\cdot]^T$ is the transpose operator; $[\cdot]^H$ is the conjugate transpose (Hermitian) operator; $\mathbf{I}_K$ stands for the $K\times K$ identity matrix; $\mathbf{X}^{-1}$ is the inverse of $X$; ${\rm det}(\cdot)$ is the determinant of a given matrix; ${\rm Re}(\cdot)$ is the real part of a complex value; $\mathcal{CN}(\mu,\sigma^2)$ defines a complex-valued Gaussian random variable with mean $\mu$ and variance $\sigma^2$; $x \sim \mathcal{U}[a,b]$ defines a uniform random variable from range of $a$ to $b$; $\left \lceil \cdot  \right \rceil$ denotes the ceiling function; $[\mathbf{X}]_{k,l}$ stands for the entry at $k$-th row and $l$-th column of matrix $\mathbf{X}$; $[\mathbf{X}]_{k,:}$ denotes the $k$-th row vector of matrix $\mathbf{X}$; $[\mathbf{X}]_{:,k}$ denotes the $k$-th column vector of matrix $\mathbf{X}$;   $[\mathbf{X}]_{a:b,c:d}$ represents the sub-matrix of $\mathbf{X}$ consisting of rows from $a$ to $b$ and columns from $c$ to $d$; The notation $\mathcal{O}(\cdot)$ denotes the computational complexity.

\section{System Model}\label{sec:sys_model}
We consider a single-cell multi-user \gls{ris}-aided \gls{m-mimo} system operating in downlink mode, where the BS is equipped with a \gls{ula} with $M_{\max}$ antennas, where $M$ antennas will be activated to serve $K$ single-antenna \gls{ue}s. {In this work, we have considered the context of far-field communication, which entails that each antenna has an equitable contribution to the system's overall performance.} The single \gls{ris} panel is equipped with $N$ elements and deployed to improve the network communication's reliability and coverage. Fig. \ref{fig:RIScell} illustrates the general \gls{ris} \gls{m-mimo} system and channel scenarios investigated.

Let us denote $\D \in \mathbb{C}^{M\times K}$ as the direct link communication channel between the \gls{ue}s and the \gls{bs}, given by
\begin{equation}
\D = \Tilde{\D} \boldsymbol{\alpha}_D^{1/2},
\end{equation}
where $\boldsymbol{\alpha}_D = \diag(\left[\alpha_{D,1},\dots,\alpha_{D, K}\right])$ and $\alpha_{D,k}$ denotes the large-scale fading 
term between the \gls{bs} and the $k$-th \gls{ue}, the matrix $\Tilde{\D} \in \mathbb{C}^{M\times K}$ is composed by independent and identically distributed (i.i.d) complex Gaussian random variables, whose mean is zero and variance is unit, $i.e.$ $d_{m,k} \sim \mathcal{CN}(0,1)$, $\forall k=1,\dots, K$ and $m=1,\dots, M$.

\begin{figure}[!ht]
\includegraphics[width=\columnwidth]{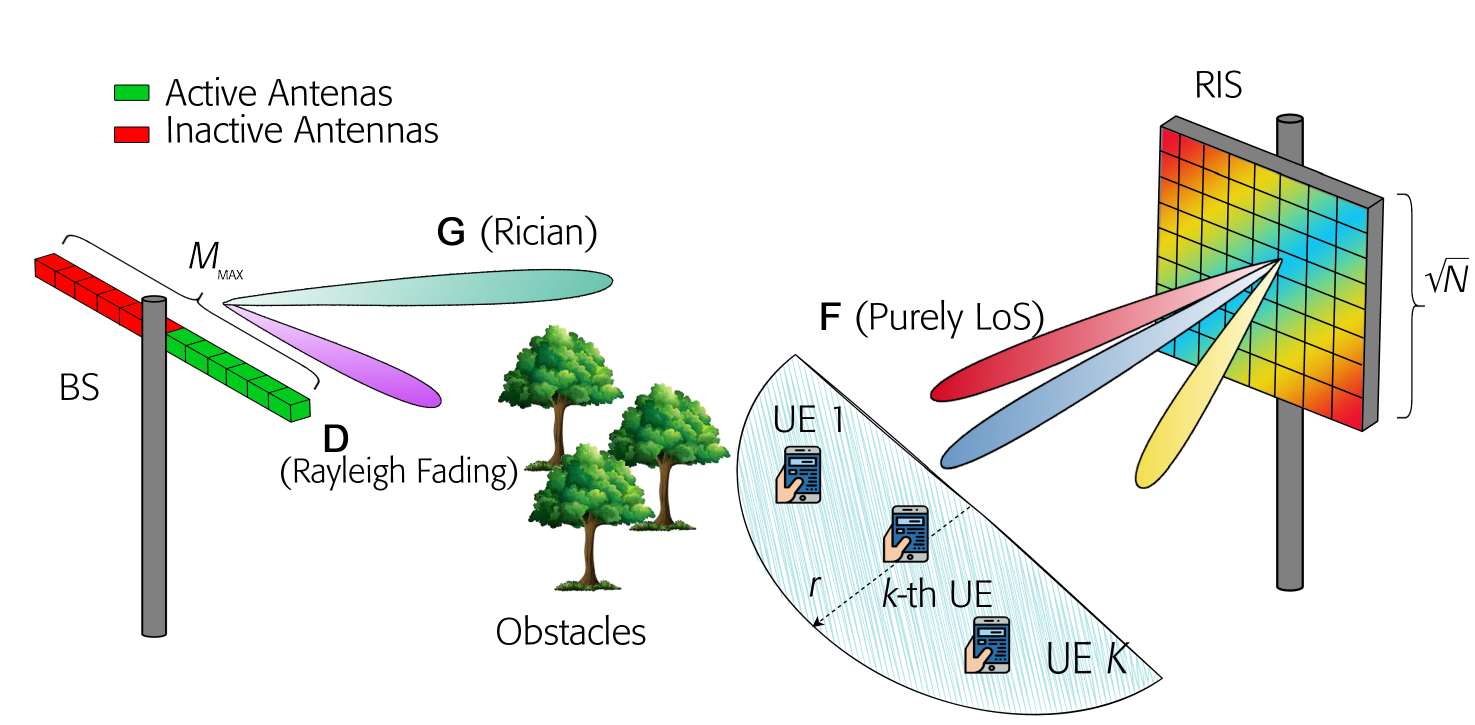}
\caption{{\gls{ris}-aided \gls{m-mimo} multi-user investigated scenarios. Particularly, this setup is illustrated with $N=64$ reflective elements at \gls{ris} and $M_{\max}=12$ \gls{bs} antennas; in this specific configuration, there are $M=6$ actives elements at \gls{bs}.}}
\label{fig:RIScell}
\end{figure}

We define the channel between the \gls{ris} and the \gls{bs} as $\mathbf{G} \in  \mathbb{C}^{M\times N}$  and the channel between the \gls{ue}s and the \gls{ris} as $\mathbf{F} \in \mathbb{C}^{N\times K}$. The channel matrix $\mathbf{G}$ is assumed to be a Rician fading channel and can be given as \cite{Zhi2022}
\begin{equation}
\mathbf{G} = \sqrt{\alpha_G\dfrac{1}{K_1 +1}} \mathbf{G}_{\rm{NLoS}} + \sqrt{\alpha_G\dfrac{K_1}{K_1+1}} \mathbf{G}_{\rm{LoS}}, 
\end{equation}
where $\alpha_G$ is the large-scale fading (path loss) between the RIS and the \gls{bs} and $K_1$ denotes the Rician factor. Besides, the fading channel matrix $\mathbf{F}$ is assumed to be purely \gls{los} and expressed as
\begin{equation}
\mathbf{F}  = \F_{\rm{LoS}} \boldsymbol{\alpha}_F^{1/2} =  [\sqrt{\alpha_{F,1}} \mathbf{f}^1_{\rm{LoS}}, \dots, \sqrt{\alpha_{F,K}} \mathbf{f}^K_{\rm{LoS}}],
\end{equation}
where $\boldsymbol{\alpha}_F = \diag([\alpha_{F,1}, \dots,\alpha_{F,K}])$ and $\alpha_{F,k}$ denotes the large-scale fading factor between the $k$-th \gls{ue} and the BS. Besides, the \gls{los} channel for \gls{ris} adopted herein is based on the two-dimensional \gls{uspa} and given as \cite{9935294} 
\begin{IEEEeqnarray}{rCl}
\mathbf{a}_{N}(\varphi,\phi) &=& \left[ 1, \dots, e^{j2\pi \frac{d}{\lambda} (x \sin\phi \sin\varphi + y\cos \phi)}, \nonumber\right. \\ 
&\dots&, \left. e^{j2\pi\frac{d}{\lambda} \left( \left( X-1 \right) \sin \phi \sin \varphi + \left( Y-1\right) \cos \phi \right)}  \right]^T,
\end{IEEEeqnarray}
where $0 \leq x,y < X$ with $X=Y\in{\mathbb{Z}_+}$ being the square side of the \gls{uspa} 
\cite{emilsquare}, \cite{near}, $d$ denotes the element spacing, $\lambda$ denotes the wavelength, $\varphi$ and $\phi$ are azimuth and elevation angles respectively. Let us denote $\varphi_{\rm{AoD}}^{bs}$ the \gls{aod} from the \gls{bs} towards the \gls{ris}, $\varphi_{{\rm AoD}}^{r,k}$ and $\phi_{\rm{AoD}}^{r,k}$ the \gls{aod} from \gls{ris} towards the $k$-th \gls{ue} and $\varphi_{{\rm AoA}}^r$ and $\phi_{{\rm AoA}}^r$ the \gls{aoa} at the \gls{ris} from the \gls{bs} respectively. Moreover, the \gls{los} channel for \gls{bs} is based on the \gls{ula}, given by \cite{8746155}
\begin{IEEEeqnarray}{rCl}
\mathbf{a}_{M}(\varphi) &=& \left[ 1, \dots, e^{j2\pi \frac{d}{\lambda} x \sin\varphi }, e^{j2\pi\frac{d}{\lambda} \left( \left( X-1 \right) \sin \varphi  \right)}  \right]^T. 
\end{IEEEeqnarray}

Under this condition, the \gls{los} channel matrix $\mathbf{G}_{{\rm LoS}}$ is expressed as 
\begin{equation}
\mathbf{G}_{{\rm LoS}} = \mathbf{a}_{M}(\varphi_{{\rm AoD}}^{bs}) \mathbf{a}_{N}(\varphi_{{\rm AoA}}^{r},\phi_{{\rm AoA}}^{r}) \triangleq  \mathbf{a}_M \mathbf{a}_N^H,
\end{equation}
and the \gls{los} component $\mathbf{f}_{{\rm LoS}}^k$ is 
\begin{equation}
\mathbf{f}_{{\rm LoS}}^k = \mathbf{a}_N(\varphi_{{\rm AoD}}^{r,k},\phi_{{\rm AoD}}^{r,k}).
\end{equation}

The phase shifts matrix of the \gls{ris} can be defined as 
\begin{equation}\label{eq:Phi}
\mathbf \Phi = \text{diag}(\mathbf{v}), \quad \mathbf{v} = [e^{-j\theta_1}, \dots, e^{-j\theta_N}]^H
\end{equation}
and $\theta_n$ is the phase shift of the $n$-th \gls{ris} element. In this way, the cascaded channel can be written as $\mathbf{H} = \mathbf{D} + \mathbf{G} \mathbf{\Phi} \mathbf{F}$.

The received signal $y_k$ at the $k$-th \gls{ue} is
\begin{equation}
    y_k = \sqrt{p_k} \mathbf{h}^H_k\mathbf{w}_k s_k + \sum_{j=1,j\neq k}^{K} \sqrt{p_{j}} \mathbf{h}_k^H \mathbf{w}_j s_j + n_k,
\end{equation}
where $\mathbf{h}_k$ and $\mathbf{w}_k$ is the $k$-th column of $\mathbf{H}$ and $\mathbf{W}$ respectively with $\mathbf{W}$ being the precoding matrix, $n_k$ is the i.i.d. additive white Gaussian noise with zero mean and unit variance, $n_k \sim \mathcal{CN}(0,1)$, $p_k \in \mathbf{p}$ is the downlink transmit power for the $k$-th \gls{ue}, where $\mathbf{p} = [p_1,\dots,p_K]^T$ denotes the set of transmit powers, and $s_k$ is the transmitted information symbol of the $k$-th \gls{ue} with $\mathbb{E}[|s_k|^2]=1$. In this work, we adopt the linear \gls{zf} precoding, 
given by the Moore–Penrose pseudoinverse $\mathbf{W} = \mathbf{H}(\mathbf{H}^H\mathbf{H})^{-1}$.

\section{Energy Efficiency Optimization Problem Formulation with Statistical CSI}\label{sec:EEoptzproblem}
Since channels might be fast time-varying in practical wireless communications scenarios, optimizing the $\boldsymbol{\Phi}$ matrix at each channel coherence time presents significant challenges, particularly because it necessitates the availability of accurate \gls{csi}; thus, utilizing the statistical \gls{csi} to optimize the \gls{ee} in \gls{ris}-aided \gls{m-mimo} can be a cost-efficient alternative way. Therefore, we formulate and investigate the \gls{ee} optimization problem in this section by exploiting the statistical \gls{csi} estimates. For this reason, we should adopt the achievable \gls{er} with \gls{zf} precoder in \gls{ris}-aided \gls{m-mimo} systems, which can be written as \cite[Eq. (18)]{ngommimo}:
\begin{equation} \label{eq:rate}
R_k = \mathbb{E}\left\{ \log_2\left(1+ \frac{p_k}{\sigma^2 \left[ (\mathbf{H}^H\mathbf{H})^{-1}\right]_{k,k}}\right)\right\}.
\end{equation}

In \cite{Zhi2022}, Kangda Zhi, $et. \; al.$, derived a closed-form lower bound expression for the \gls{er} as given in Eq. \eqref{eq:rate}, in a \gls{ris}-aided \gls{m-mimo} system with \gls{zf}. This expression can be represented as \cite[Eq. (22)]{Zhi2022}
\begin{equation} \label{eq:rateLB}
R_k \geq \widetilde{R}_k =  \log_2\left(1 + \frac{p_k(M-K) \mathbf{v}^H \mathbf{B} \mathbf{v} }{\sigma^2(K_1 + 1) \mathbf{v}^H \mathbf{A}_k \mathbf{v}}    \right), 
\end{equation}
where the constants $\mathbf{B}$ and $\mathbf{A}_k$ $\forall k\in \{1,\dots,K\}$ are given as
\begin{equation} \label{eq:B}
    \mathbf{B} \triangleq \frac{1}{N} \mathbf{I}_N + \alpha_G K_1 \; \diag(\mathbf{a}_N^H) \F \mathbf{\Lambda}^{-1} \F^H \diag(\mathbf{a}_N)
\end{equation}

\begin{equation} \label{eq:Ak}
    \mathbf{A}_k \triangleq \left[ \mathbf{\Lambda}^{-1}\right]_{k,k} \mathbf{B} - K_1 \alpha_G \mathbf{s}_k \mathbf{s}_k^H,
\end{equation}
where the variables $\mathbf{\Lambda} \triangleq \alpha_G \F^H \F + (K_1+1)\diag(\boldsymbol{\alpha}_D),$ and $\mathbf{s}_k^H \triangleq \left[ \mathbf{\Lambda^{-1}} \F^H \diag(\mathbf{a}_N)\right]_{k,:}$ are computed through of the path-loss and \gls{los} component 

Thus, the {\it maximum \gls{ee} optimization} problem regarding the interest variables is described as follows.
\begin{IEEEeqnarray}{rCl} \label{eq:P1}
 & \Pn: \, \underset{{M,\mathbf{p},\mathbf{v}} }{{\rm maximize}} &\;\; \etaee =   
\frac{ \displaystyle \sum_{k=1}^K \log_2\left(1 + \frac{p_k(M-K) \mathbf{v}^H \mathbf{B} \mathbf{v} }{\sigma^2(K_1 + 1) \mathbf{v}^H \mathbf{A}_k \mathbf{v}}    \right)}{\varrho \displaystyle \sum_{k=1}^K p_k + P_{\rm{FIX}} + M P_{\rm{BS}} + N P_{\rm{RIS}}} 
    \nonumber \\\\
& {\rm subject\;to} & \hspace{4mm}  \displaystyle \sum_{k=1}^K p_k \leq \ptx, \IEEEyessubnumber \label{c:1}
    \\
&& \hspace{4mm} \widetilde{R}_k \geq R_{\min}, \;\;\;\forall k=1,\dots,K, \IEEEyessubnumber \label{c:2}
    \\
&& \hspace{4mm} M > K,  \IEEEyessubnumber \label{c:3}
    \\
&& \hspace{4mm}  |e^{j \theta_n }| = 1, \;\;\; \forall n=1,\dots,N, \IEEEyessubnumber \label{c:4}
\stepcounter{problem}
\end{IEEEeqnarray}
where $\varrho$ is the \gls{pa} inefficiency; $\pbs$ is the constant \gls{rf} chain circuit power consumption per transmit antenna; $\pris$ is a fixed power value necessary to activate one reflecting element on \gls{ris}; $P_{{\rm FIX}}$ is the fixed consumed power at the \gls{bs} \cite{EEdebbah,9497709}. Constraint \eqref{c:1} is the total \gls{rf} transmit power constraint at \gls{bs} with $\ptx$ being the {\it power budget} available at \gls{bs}, $i.e.$, the maximum power available; constraint \eqref{c:2} guarantees the minimum achievable downlink rate for all \gls{ue}s, where $R_{\min}$ is a pre-defined parameter which depends on \gls{qos} requirements; constraint \eqref{c:3} guarantees the applicability of the \gls{zf} in the \gls{er} lower bound equation \footnote{In the practice, it is important to note that $M\geq K$ for \gls{zf} feasibility, however, when considering the \gls{er} defined in \eqref{eq:rateLB}, the number of antennas should exceed than the number of \gls{ue}s, $i.e.$, $M>K$, as the rate becomes null when the equality condition is met.}, and constraint \eqref{c:4} accounts for the fact that each \gls{ris} reflecting element can only redirect, without amplifying the incoming signal, $i.e.$, the \gls{ris} is assumed to operate exclusively only in passive mode.

The \textbf{\textit{energy consumption}} model builds upon significant contributions from established works in the literature, particularly
\cite{EEdebbah, 9497709}. Our primary emphasis in this work is to demonstrate the efficiency and effectiveness of the developed optimization methodology in yielding \gls{ris}-aided \gls{m-mimo} \gls{ee} performance enhancements at the cost of an affordable computational complexity. It is worth noting that our chosen energy consumption model inherently accounts key factors such as individual antenna element consumption ($\pbs$) and \gls{ris} consumption per element ($\pris$). This consideration is sufficient in attaining our set of objectives.

We can see that the objective function of the problem given in Eq. \eqref{eq:P1} is a concave-linear fraction whose numerator and denominator are concave and linear functions w.r.t. the number of active antennas $M$ and power allocation $\mathbf{p}$ respectively. Furthermore, one can notice that the constraint \eqref{c:2} is a non-linear constraint; and the constraint \eqref{c:4} is non-convex; thus, the formulated problem in \eqref{eq:P1} is an intricate non-convex \gls{nlp}, hence, difficult to solve.  

Due to the coupling of variables, it is hard to obtain the globally optimal solution of \eqref{eq:P1}. Therefore, we proposed an iterative procedure 
to solve problem \hyperref[eq:P1]{$\mathcal{P}_1$} sub-optimally and with low complexity. For this purpose, we leverage the classical {\it alternating optimization} strategy, which optimizes one variable while the others are kept fixed. The proposed alternating algorithm contains basically two fundamental steps, which are run sequentially. In the first step, we optimize the \gls{ris} phase shifts $\mathbf{v}$ for a fixed number of active antennas $M$ and power allocation $\mathbf{p}$; in the second step, we optimize the number of active antennas $M$ and power allocation $\mathbf{p}$ with given \gls{ris} phase shifts $\mathbf{v}$. 

\subsection{Phase Optimization} \label{ss:Phase}

In this subsection, we deal with the \gls{ris} phase shifts optimization problem, where we design a new low-complexity algorithm 
based on an analytical approach. 
Specifically, we first apply the 
\gls{ldt} procedure \cite{8310563,8982186} to translate our sum-of-logarithm original problem into an equivalent sum-of-ratios problem. In addition, we utilized the 
strategy to deal with the sum-of-ratio proposed in \cite{Jong2012AnEG} to reach a closed-form expression for $\theta_n$ $\forall n \in \{1,\dots,N\}$.

\subsubsection{Lagrangian Dual Transform} \label{subsubsec:LDT}

From the objective function of \eqref{eq:P1}, we observe that when $\mathbf{p}$ and $M$ are fixed, the \gls{ee} maximization is equivalent to the sum-rate maximization; therefore, the optimal $\mathbf{v}$ is the one that maximizes the numerator of \eqref{eq:P1}. One can notice that due to the sum-of-logarithm, the closed-form expression is challenging or even impracticable to obtain directly; to address this issue, we initially employed the {\it Lagrangian Dual Transform} proposed in 
\cite{8310563} to decouple the argument of the logarithm function argument, enabling us to derive the equivalent optimization problem, given as:
\begin{IEEEeqnarray}{rCl} \label{eq:P2}
 \Pn: \;\;\; && \underset{\mathbf{v},\boldsymbol{\gamma}}{\rm maximize} \quad \displaystyle   \sum_{k=1}^K  \left(  \vphantom{\frac{(1+\gamma_k) p_k(M-K) \mathbf{v}^H \mathbf{B} \mathbf{v}}{\mathbf{v}^H \left( \sigma^2(K_1 + 1) \mathbf{A}_k   + p_k(M-K)  \mathbf{B} \right )\mathbf{v} } }  \log_2(1 + \gamma_k)  -\gamma_k  \right.
 \\
 && \left.  + \; \frac{(1+\gamma_k) p_k(M-K) \mathbf{v}^H \mathbf{B} \mathbf{v}}{\mathbf{v}^H \left( \sigma^2(K_1 + 1) \mathbf{A}_k   + p_k(M-K)  \mathbf{B} \right )\mathbf{v} } \right), \nonumber
\\ \nonumber
\\
&& {\rm subject\; to} \quad  \gamma_k \geq 0, \quad \forall k=1,\dots,K, \nonumber
\\
&& \qquad \qquad \quad \; |e^{j \theta_n }| = 1, \;\;\; \forall n=1,\dots,N \nonumber.
\stepcounter{problem}
\end{IEEEeqnarray}
where $\boldsymbol{\gamma} = [\gamma_1,\dots,\gamma_K]^T$ is an auxiliary variable 
introduced for each argument of the logarithm function $\frac{p_k(M-K) \mathbf{v}^H \mathbf{B} \mathbf{v} }{\sigma^2(K_1 + 1) \mathbf{v}^H \mathbf{A}_k \mathbf{v}}$in the numerator of Eq. \eqref{eq:P1}. Here, the idea is optimize $\boldsymbol{\gamma}$ and $\mathbf{v}$ in an alternating, interative way, $i.e.$, we optimize $\boldsymbol{\gamma}$ for given $\mathbf{v}$, while $\mathbf{v}$ is updated for values of $\boldsymbol{\gamma}$ at their last updated. The optimal values of $\boldsymbol{\gamma}$ can be straightly obtained in closed-form according to the \gls{kkt} conditions applied to Eq. \eqref{eq:P2}, resulting: 
\begin{equation} \label{eq:gammaopt}
\overset{\hspace{-1.3mm}\customstar}{\gamma_k} = \left[ \left( \frac{\sigma^2(K_1 + 1) \mathbf{v}^H \mathbf{A}_k \mathbf{v} + p_k(M-K) \mathbf{v}^H \mathbf{B} \mathbf{v} }{\log(2)\sigma^2(K_1 + 1) \mathbf{v}^H \mathbf{A}_k \mathbf{v}} \right) -1  \right]^+.
\end{equation}

Given $\overset{\hspace{0.4mm}\customstar}{\boldsymbol{\gamma}}$, one may focus on the following fractional problem 
\begin{equation} \label{eq:sum-of-ratio}
    \mathcal{P}_2^{'}:  \, \underset{\mathbf{v}}{\rm maximize} \, \sum_{k=1}^K \frac{(1+\overset{\hspace{-1.3mm}\customstar}{\gamma_k}) p_k(M-K) \mathbf{v}^H \mathbf{B} \mathbf{v}}{\mathbf{v}^H \left( \sigma^2(K_1 + 1) \mathbf{A}_k   + p_k(M-K)  \mathbf{B} \right )\mathbf{v} }.
\end{equation}

\subsubsection{Fractional Problem Sum-of-Ratio}\label{sec:fracSoR}
{As the problem $\mathcal{P}_2^{'}$ given by Eq. \eqref{eq:sum-of-ratio} involves} a sum-of-ratios, classical transformations for {single-ratio} fractional programming problems, {such} as \textit{Charnes-Cooper Transform} and \textit{Dinkelbach's Transform} {cannot be easily generalized} to multiple-ratio problems \cite{8314727}. {Furthermore, while} the \textit{Quadract Transform} proposed in \cite{8314727} {performs well} in many multiple-ratio problems, {particularly in the problem of phase shifts optimization, it can pose challenges in achieving our primary objective, which is} obtaining a closed-form expression {for each element of $\mathbf{v}$}, as it introduces a square root in the optimization variable. To circumvent this problem, herein we invoke the strategy proposed in \cite{Jong2012AnEG}. 

The key idea {for the method proposed in \cite{Jong2012AnEG}}, is the introduction of auxiliary variables $\boldsymbol{\beta} = [\beta_1,\dots,\beta_K]^T$, {which enables an equivalent representation of the} problem represented by Eq. \eqref{eq:sum-of-ratio}, as the following

\begin{IEEEeqnarray}{rCl} \label{eq:p3'}
\Pn 
: \; && \underset{\mathbf{v},\boldsymbol{\beta}}{\rm maximize}  \sum_{k=1}^K u_k\left( \vphantom{\sum}(1+\overset{\hspace{-1.3mm}\customstar}{\gamma_k})p_k(M-K)\mathbf{v}^H\mathbf{B} \mathbf{v} \right. \nonumber
    \\
    && \left. - \;
    \beta_k \mathbf{v}^H \left( \sigma^2(K_1 + 1) \mathbf{A}_k   + p_k(M-K)  \mathbf{B} \right )\mathbf{v} \vphantom{\sum} \right)  
    \\
    \nonumber
    \\
    && {\rm subject\; to} \quad  \beta_k \geq 0, \quad \forall k=1,\dots,K, \nonumber
\\
&& \qquad \qquad \quad \; |e^{j \theta_n }| = 1, \;\;\; \forall n=1,\dots,N \nonumber.
\end{IEEEeqnarray}
where 
\begin{equation} \label{eq:u}
    u_k = \frac{1}{\mathbf{v}^H \left( \sigma^2(K_1 + 1) \mathbf{A}_k   + p_k(M-K)  \mathbf{B} \right )\mathbf{v}},
\end{equation}
is defined as the denominator of the objective function given in Eq. \eqref{eq:sum-of-ratio}, $\forall k=1,\dots, K$, and is set constant being updated in each interaction after $\mathbf{v}$ is optimized. Besides, according to \cite{Jong2012AnEG} the optimal values of $\boldsymbol{\beta}$ can be obtained as:

\begin{equation} \label{eq:betaopt}
\overset{\hspace{-0.9mm}\customstar}{\beta_k} = \frac{(1+\overset{\hspace{-1.3mm}\customstar}{\gamma_k}) p_k(M-K) \mathbf{v}^H \mathbf{B} \mathbf{v}}{\mathbf{v}^H \left( \sigma^2(K_1 + 1) \mathbf{A}_k   + p_k(M-K)  \mathbf{B} \right )\mathbf{v} }.
\end{equation}

Notice that for given $\overset{\hspace{0.5mm}\customstar}{\boldsymbol{\beta}}$, we can rewrite Eq. \eqref{eq:p3'} as

\begin{IEEEeqnarray}{rCl} \label{eq:P3}\label{eq:quadraticproblem}
    \Pn^{'}: \quad && \underset{\mathbf{v}}{\rm maximize} \quad  \mathbf{v}^H \mathbf{C} \mathbf{v},
    \\
    && {\rm subject\; to}  \qquad  |e^{j \theta_n }| = 1, \;\;\; \forall n=1,\dots,N \nonumber
    \stepcounter{problem}
\end{IEEEeqnarray}

where $\mathbf{C} \in \mathbb{C}^{N \times N}$ is a Hermitian matrix defined as

\begin{IEEEeqnarray}{rCl}  \label{eq:C}
 \mathbf{C} = && \sum_{k=1}^K u_k\left( \vphantom{\sum} (1+\overset{\hspace{-1.3mm}\customstar}{\gamma_k})p_k(M-K)\mathbf{B}  \right.
\\
&& \left. - \overset{\hspace{-0.9mm}\customstar}{\beta_k}  \left( \sigma^2(K_1 + 1) \mathbf{A}_k   + p_k(M-K)  \mathbf{B} \right ) \vphantom{\sum} \right).\nonumber 
\end{IEEEeqnarray}

\begin{figure*}[!b] 
\hrulefill
\normalsize
\begin{IEEEeqnarray}{rCl}\label{eq:lagrangian}
\mathcal{L}\left(M,\mathbf{\widetilde{R
}},\boldsymbol{\mu}, \vartheta \right) =  &\sum_{k=1}^K& \widetilde{R}_k
 - \etaee\left(\varrho \displaystyle \sum_{k=1}^K \frac{(2^{\widetilde{R}_k}-1)\sigma^2(K_1+1)\mathbf{v}^H\mathbf{A}_k\mathbf{v}}{(M-K)\mathbf{v}^H\mathbf{B}\mathbf{v}} + P_{\rm{FIX}} +  M P_{\rm{BS}} +  N P_{\rm{RIS}}\right) + \nonumber \\
 && \displaystyle + \sum_{k=1}^K \mu_k \left(\widetilde{R}_k -  R_{\min} \right)  
+\vartheta \left(\ptx -  \sum_{k=1}^K \frac{(2^{\widetilde{R}_k}-1)\sigma^2(K_1+1)\mathbf{v}^H\mathbf{A}_k\mathbf{v}}{(M-K)\mathbf{v}^H\mathbf{B}\mathbf{v}}\right). 
\end{IEEEeqnarray}
\vspace*{4pt}
\end{figure*}

The problem \hyperref[eq:P3]{$\mathcal{P}_3$}  
is known in the literature, and initially, the first proposed solution has been the \gls{sdr} \cite{5447068}. {Nevertheless, the performance of the \gls{sdr} algorithm is highly reliant on Gaussian randomization, particularly when the rank of the matrix ($\mathbf{v}\mathbf{v}^H$) exceeds one. On the one hand, \gls{sdr} can introduce additional implementation complexity, primarily attributed to the time-consuming nature of solving the \gls{sdr} problem, especially as the matrix dimensions increase. Furthermore, achieving an improved solution often necessitates a substantial number of randomization. Therefore, this motivates us to develop two different strategies:
a) Analytical approach, where we derived innovative closed-form expressions, and 
b) \gls{sfp} {method}, also known as the \gls{mm} strategy.}

{In pursuit of developing a low-complexity algorithm, we have successfully derived a novel closed-form expression for the phase shifts of \gls{ris}. This achievement is encapsulated in the following lemma:}

\begin{lemma} \label{lemma:1}
    Let the optimization problem given in Eq. \eqref{eq:quadraticproblem}, when $v_{\ell}=e^{j\theta_\ell}$ is fixed, $\forall \ell \in \{1,\dots,N\} \setminus n$, the optimal solution of $\overset{\hspace{-2mm}\customstar}{[\mathbf{v}]_n}$ can be given in closed-form expression by

    \begin{IEEEeqnarray}{rCl} \label{eq:thetaopt}
    \overset{\hspace{-1.4mm}\customstar}{\theta_n} &=& \angle \left[ {\rm Tr}\left( \left(\left[\mathbf{C}\right]_{1:(n-1),n}\right)^H \left[\mathbf{v}\right]_{1:(n-1)} \right) \right.\nonumber  \\
     &&\qquad \qquad \left. + \; {\rm Tr}\left( \left[\mathbf{v}\right]_{(n+1):N} \left[\mathbf{C}\right]_{n,(n+1):N}  \right)  \vphantom{\left(\left[\mathbf{C}\right]_{1:(n-1),n}\right)^H}  \right  ], \nonumber \\
    && \qquad \qquad \qquad \qquad \forall n \in \{1,\dots,N\}
\end{IEEEeqnarray}
\end{lemma}

\begin{proof}
The proof is Available in Appendix \ref{app:lemma}.
\end{proof}

For the \gls{sfp} method, we employed the strategy described in \cite[Lemma 2]{EEdebbah}. The proposed {\gls{fp}-based \gls{ris} phase shift optimization}, based on statistical-\gls{csi} for \gls{ris}-aided \gls{m-mimo}, is summarized in Algorithm \ref{alg:phase-optz}. 
The \gls{ris} phase shifts $\mathbf{v}$, are optimized when the \gls{ue}'s power allocation $\mathbf{p}$ and the number of transmit antennas at \gls{bs} $M$ are fixed. 

\begin{algorithm}  
\caption{\gls{fp}-based \gls{ris} Phase shifts Optimization}
\begin{algorithmic}
\State \textbf{Input:} $M$, $\mathbf{p}$, $\mathbf{v}$, $\sigma^2$, $K_1$, $\mathbf{B}$, $\mathbf{A}_k$, $\forall k$ $\in \{1,\dots,K\}$
\Repeat
\State \textbf{Step 1:} Update $\boldsymbol{\gamma}$ by \eqref{eq:gammaopt};
\Repeat
\State \textbf{Step 2:}  Compute $\mathbf{u}$ by \eqref{eq:u}; 
\State \textbf{Step 3:}  Update $\boldsymbol{\beta}$ by \eqref{eq:betaopt}; 
\State \textbf{Step 4:} Compute  $\mathbf{C}$ by \eqref{eq:C};
\State  Set $\ell = 1$;
\Repeat
\State \vspace{-0.35cm}
\State \textbf{Step 5:}
\vspace{0.05cm}
\State  {\it 1) Analytical Approach:}
\State \hspace{0.2cm}  $[\mathbf{v}]_{n}^{(\ell)} = e^{j\overset{\hspace{-1.3mm}\customstar}{\theta_n}}$ based on \eqref{eq:thetaopt}, $ n=1,\dots,N$;
\State \vspace{-0.2cm}
\State  {\it 2) SFP Approach:}
\State \hspace{0.2cm}  Set $\lambda_{\min}$ as the minimum eigenvalue of $\mathbf{C}$;
\State \hspace{0.2cm}  $\mathbf{v}^{(\ell)} = e^{j\angle(\lambda_{\min}\mathbf{I}_N-\mathbf{C})\mathbf{v}^{(\ell-1)}}$;
\Until $|| \mathbf{v}^{(\ell)} - \mathbf{v}^{(\ell-1)} || < \epsilon$
\Until the objective function in \eqref{eq:sum-of-ratio} converge
\Until the objective function in \eqref{eq:P2} converge
\State \textbf{Output:} $\overset{\hspace{0mm}\customstar}{{\bf v}} = [e^{-j\overset{\hspace{-1.1mm}\customstar}{\theta_1}},\dots,e^{-j\overset{\hspace{-1.75mm}\customstar}{\theta_N}}]^H$
\end{algorithmic}
\label{alg:phase-optz}
\end{algorithm}

\subsection{Active Antennas  and Power Allocation Optimization}

The formulated original problem \hyperref[eq:P1]{$\mathcal{P}_1$} is a 
fractional programming problem, whereas Dinkelbach's transform is a classical method that can deal with this kind of problem. By invoking the Dilkelbach's transform, problem \hyperref[eq:P1]{$\mathcal{P}_1$} can be transformed to
\begin{IEEEeqnarray}{rCl} \label{eq:P4}
 \Pn: && \; \underset{M,\mathbf{p}}{\rm maximize} \displaystyle \sum_{k=1}^K \log_2 \left(1 + \frac{p_k(M-K) \mathbf{v}^H \mathbf{B} \mathbf{v} }{\sigma^2(K_1 + 1) \mathbf{v}^H \mathbf{A}_k \mathbf{v}}    \right) 
    \\
&& -\eta_{{\rm EE}} \left( \varrho \displaystyle \sum_{k=1}^K p_k + P_{\rm{FIX}} + M P_{\rm{BS}} + N P_{\rm{RIS}} \right).  \nonumber
\\
&&{\rm subject\;to} \quad \eqref{c:1},  \eqref{c:2} \nonumber.
\stepcounter{problem}
\end{IEEEeqnarray}

Notice that the existence of logarithmic functions imposes a serious difficulty in deriving a closed-form solution for $M$. Intending to propose a low complexity algorithm for optimization variable $M$, we proceed similarly to \cite{8457308}. Hence, by recalling Eq. \eqref{eq:rateLB} and after some mathematical manipulations we can define
\begin{equation}\label{eq:pk}
p_k = \frac{(2^{\widetilde{R}_k}-1)\sigma^2(K_1+1)\mathbf{v}\mathbf{A}_k\mathbf{v}}{(M-K)\mathbf{v}\mathbf{B}\mathbf{v}}
\end{equation}

Substituting \eqref{eq:rateLB} and \eqref{eq:pk} into \hyperref[eq:P4]{$\mathcal{P}_4$} we can further convert this problem to the following equivalent problem:
\begin{IEEEeqnarray}{rCl} \label{eq:P5}
 \Pn: && \; \underset{M,\widetilde{\mathbf{R}}}{\rm maximize} \displaystyle \sum_{k=1}^K \widetilde{R}_k -\etaee\left( \vphantom{\displaystyle \sum_{k=1}^K \frac{(2^{\widetilde{R}_k}-1)\sigma^2(K_1+1)\mathbf{v}^H\mathbf{A}_k\mathbf{v}}{(M-K)\mathbf{v}^H\mathbf{B}\mathbf{v}} } P_{\rm{FIX}} + M P_{\rm{BS}} + N P_{\rm{RIS}}  \right.
  \nonumber  \\
&& \left.  \varrho \displaystyle \sum_{k=1}^K \frac{(2^{\widetilde{R}_k}-1)\sigma^2(K_1+1)\mathbf{v}^H\mathbf{A}_k\mathbf{v}}{(M-K)\mathbf{v}^H\mathbf{B}\mathbf{v}}  \right). \\
&& {\rm subject\;to}  \; \displaystyle \sum_{k=1}^K \frac{(2^{\widetilde{R}_k}-1)\sigma^2(K_1+1)\mathbf{v}^H\mathbf{A}_k\mathbf{v}}{(M-K)\mathbf{v}^H\mathbf{B}\mathbf{v}}  \leq \ptx, \nonumber\\ \IEEEyessubnumber \\
&& \hspace{20mm} \widetilde{R}_k \geq R_{\min}, \;\;\;\forall k=1,\dots,K \IEEEyessubnumber \hspace{12.5mm} \eqref{c:2} \nonumber
\end{IEEEeqnarray}

Since  $\mathbf{v}^H\mathbf{A}_k\mathbf{v} > 0$ and $\mathbf{v}^H\mathbf{B}\mathbf{v}>0$, \hyperref[eq:P5]{$\mathcal{P}_5$} is convex, one can utilize the Lagrangian Dual method \cite{boyd} to solve this problem and reach near-optimal low-complexity solution regarding the original problem \hyperref[eq:P1]{$\mathcal{P}_1$}.  

The Lagrangian Dual problem can be modeled by \cite{boyd,tang}
\begin{IEEEeqnarray}{rCl}
&\underset{\boldsymbol{\mu},\vartheta}{{\rm minimize}} 
\quad  \underset{M,\mathbf{\widetilde{R}}}{\rm{maximize}} 
\quad
\mathcal{L}\left(M,\mathbf{\widetilde{R}}, \boldsymbol{\mu}, \vartheta \right),&
\\
&\hspace{-2cm}{\rm subject \; to} \quad \boldsymbol{\mu} \geq  \mathbf{0}, \,\,\vartheta \geq 0,&
\end{IEEEeqnarray}
where the corresponding Lagrangian function can be expressed as \eqref{eq:lagrangian}, in which the vector $\boldsymbol{\mu}$ and the scalar $\vartheta$ are non-negative Lagrange multipliers.

We can obtain the optimal number of transmit antennas $\overset{\hspace{0.5mm}\customstar}{M}$ by satisfying the \gls{kkt} conditions, which can be directly computed as: 
\begin{IEEEeqnarray}{rCl}\label{eq:KKTantenna}
     && \frac{\partial \mathcal{L}(M,\widetilde{\mathbf{R}},\boldsymbol{\mu},\vartheta)}{\partial M} = 0, \quad \frac{\partial \mathcal{L}(M,\widetilde{\mathbf{R}},\boldsymbol{\mu},\vartheta)}{\partial \widetilde{R}_k} = 0.
\end{IEEEeqnarray}

After some mathematical manipulations  from the computed derivatives via Eq. \eqref{eq:KKTantenna}, the following second-order equation can be obtained:
\begin{IEEEeqnarray}{rCl}
&& -(M-K)^2 \etaee \pbs\log(2) + (M-K)\sum_{k=1}^K(\mu_k+1) \nonumber \\ 
&&    -(\etaee\varrho+\vartheta)\sigma^2(K_1+1)\log(2)\sum_{k=1}^K \frac{\mathbf{v}^H\mathbf{A}_k\mathbf{v}}{\mathbf{v}^H\mathbf{B}\mathbf{v}} = 0
\end{IEEEeqnarray}

\begin{figure*}[!t]
\begin{equation}\label{eq:Mopt}
\overset{\hspace{0.5mm}\customstar}{M} = \left[\frac{ \scalebox{0.8}{$\displaystyle \sum_{k=1}^K$}  (\mu_k+1) + \sqrt{\left(\scalebox{0.8}{$\displaystyle \sum_{k=1}^K$}(\mu_k+1)\right)^2 - 4\etaee\pbs\log(2)^2\sigma^2(K_1+1) (\etaee \varrho+\vartheta) \scalebox{0.8}{$\displaystyle \sum_{k=1}^K$} \frac{\mathbf{v}\mathbf{A}_k \mathbf{v} }{\mathbf{v}\mathbf{B} \mathbf{v}}  }}{2\etaee\pbs\log(2)} + K\right]^{\scalebox{.9}{$M_{\max}$}}_{\scalebox{.9}{$K+1$}}, 
\end{equation}
\hrule
\end{figure*}
Finally, the optimal number of transmit antennas $\overset{\hspace{0.5mm}\customstar}{M}$ in closed-form is given at the top of this page by Eq. \eqref{eq:Mopt}.

Notice that for the proposed solution be feasible in practice, $\overset{\hspace{0.5mm}\customstar}{M}$ must be limited in the range of $K+1$ and $M_{\max}$ since the constraint of the \gls{zf} precoding, \eqref{c:3}, should be attained and $M$ should be lower than the maximum number of antennas available at the \gls{bs}, $M_{\max}$.

Proceeding by substituting the calculated derivatives, Eq. \eqref{eq:KKTantenna}, into the power equation $p_k$, Eq. \eqref{eq:pk}, we obtain the optimal power allocation closed-form expression:
\begin{equation} \label{eq:popt}
\overset{\hspace{-1.8mm}\customstar}{p_k}  = \left[ \dfrac{1+\mu_k}{\log(2)(\etaee\varrho +\vartheta)} -\frac{\sigma^2(K_1+1)}{M-K}\dfrac{\mathbf{v}^H\mathbf{A}_k\mathbf{v}}{\mathbf{v}^H\mathbf{B}\mathbf{v}} \right]^+.
\end{equation}

\vspace{2mm}
\noindent\textbf{\textit{Lagrange multipliers Updating}}. Although the analytical expressions for the optimal number of transmit antenna, $\overset{\hspace{0.5mm}\customstar}{M}$, and power allocation, $\overset{\hspace{-1.8mm}\customstar}{p_k} $, have been derived, these expressions are related to the Lagrangian multipliers. Herein, we apply the widely-used sub-gradient method to update the Lagrangian multipliers; it means that $\mu_k$ and $\vartheta$ should decrease if the gradients are positives, {\it i.e.}, $\nabla_{\mu_k} \mathcal{L} > 0$ and $\nabla_{\vartheta} \mathcal{L} > 0$, and vice versa. The values of $\mu_k$ and $\vartheta$ at the $\ell$-th iteration  
will be
updated according to
\begin{equation} \label{eq:multipliers}
\begin{aligned}
     \mu_k^{(\ell)} &= \left[ \mu_k^{(\ell-1)} - \frac{\left(R_k - R_{\min} \right)}{10}    \right]^+
    \\
    \vartheta^{(\ell)} &= \left[ \vartheta^{(\ell-1)} - \frac{\sqrt{\ell}}{10}  \left(\ptx - \sum_{k=1}^K p_k \right)  \right]^+,
\end{aligned}
\end{equation}

\begin{algorithm}   
\caption{Proposed Solution for $M$ and $\mathbf{p}$ Optimization}
\begin{algorithmic}
\State \textbf{Input:} $M$, $\sigma^2$, $K_1$, $\mathbf{p}$, $\mathbf{v}$, $\mathbf{B}$, $\mathbf{A}_k$, $\forall k$ $\in \{1,\dots,K\}$
\Repeat
\State \textbf{Step 0:} Compute $\etaee$ by \eqref{eq:P1}; 
\State \textbf{Step 1:} Compute  $M$ by \eqref{eq:Mopt};
\State \textbf{Step 2:} Compute $\mathbf{p}$ by \eqref{eq:popt};
\State \textbf{Step 3:} Update the Lagrangian Multipliers by \eqref{eq:multipliers};
\Until $\etaee$ converge
\State \textbf{Outputs:} $\overset{\hspace{0.5mm}\customstar}{M}$ and $\overset{\hspace{0mm}\customstar}{\mathbf{p}}$
\end{algorithmic}
\label{alg:p-M-optz}
\end{algorithm}

\begin{algorithm}  
\caption{\gls{ee} Maximization Complete Solution Algorithm}
\begin{algorithmic}
\State \textbf{Input:}  $N$, $\ptx$, $\sigma^2$, $K_1$, $\alpha_G$, $\boldsymbol{\alpha}_F$, $\boldsymbol{\alpha}_D$, $\mathbf{a}_N$, $\mathbf{F}$ 
\State  \textbf{Step 0:}  Initialize $M$, $\mathbf{p}$, and $\mathbf{v}$ to feasible values;
\State \textbf{Step 1:} Compute $\mathbf{B}$ and $\mathbf{A}_k$ as \eqref{eq:B} and \eqref{eq:Ak} respectively;
\Repeat
\State \textbf{Step 2:} Compute $\mathbf{v}$ according to Algorithm \ref{alg:phase-optz};
\State \textbf{Step 3:} Update $\mathbf{A}_k$ as \eqref{eq:Ak};
\State \textbf{Step 4:} Compute $\mathbf{p}$ and $M$ by Algorithm \ref{alg:p-M-optz}; 
\Until the objective function in \eqref{eq:P1} converge
\State {\bf Outputs:} $\overset{\hspace{0.5mm}\customstar}{M}$, $\overset{\hspace{0mm}\customstar}{\mathbf{p}}$ and $\overset{\hspace{0mm}\customstar}{{\bf v}}$
\end{algorithmic}
\label{alg:complete} 
\end{algorithm}

\subsection{Complexity}\label{sec:complexity}

We provide the computational complexity analysis of the proposed Algorithms \ref{alg:phase-optz}, \ref{alg:p-M-optz}, and \ref{alg:complete}. Concerning to the Algorithm \ref{alg:phase-optz}, it is known that the complexity to update   $\boldsymbol{\gamma}$, $\mathbf{u}$, and $\boldsymbol{\beta}$ results in 
the same order, given by $\mathcal{O}(KN^2)$. Additionally, solve \hyperref[eq:P3]{$\mathcal{P}_3$} by Eq. \eqref{eq:thetaopt} has the complexity of
$C^{\rm ANA} = \mathcal{O}(N(N-1))$, while solve \hyperref[eq:P3]{$\mathcal{P}_3$} with \gls{sfp} methodology has complexity of $C^{\rm SFP} = \mathcal{O}(N)$. Furthermore, Algorithm \ref{alg:p-M-optz} has a complexity of $\mathcal{O}(3I_2KN^2)$. Therefore, the complexity of the proposed \gls{ris} phase shifts optimization $\mathbf{v}$ procedure, as well as for the number of active antennas $M$ and power allocation $\mathbf{p}$ optimization, and the complete proposed solution procedure are given in Table \ref{tab:complexity}.

\begin{table}[!ht]
\caption{Complexity for the Optimization Method}
\centering
\begin{tabular}{c|c}
\hline
\bf Alg.   & \bf Complexity  \\
\hline
\bf \ref{alg:phase-optz}   & \small $\mathcal{O}(I_1(KN^2 + \hat{I}_1(2KN^2+\bar{I}_1 C^{\{i\}})))$\\
\bf \ref{alg:p-M-optz}   & \small $\mathcal{O}(3I_2KN^2)$ \\
\bf \ref{alg:complete}  &  \small $\mathcal{O}(I_3( (I_1(KN^2 + \hat{I}_1(2KN^2+\bar{I}_1 C^{\{i\}})))  + 3I_2KN^2))$\\
\hline
\end{tabular}
\label{tab:complexity}
\end{table}
\noindent where $I_1$ and $\hat{I}_1$ represent the number of iterations for the inner and outer layers of Algorithm \ref{alg:phase-optz}, while  
$\bar{I}_1$ is the number of iterations for optimizing $\mathbf{v}$, and $i \in \{\textsc{sfp}, \textsc{ana}\}$ holds for \gls{sfp} and Analytical approach, respectively; $I_2$ is the number of iterations for Algorithm \ref{alg:p-M-optz} and $I_3$ is the number of iterations for Algorithm \ref{alg:complete}.

\section{Simulation Results}\label{sec:simul}
In this section, we aim to investigate the performance of the multi-user \gls{ris}-aided \gls{m-mimo} system. We denote a system configuration setup where the \gls{ue}s' localization is  
fixed concerning their positions, \gls{aoa}, and \gls{aod} for azimuth/elevation as illustrated in Fig. \ref{fig:RIScell}. For one single system configuration setup, the instantaneous-\gls{csi} optimization-based approach is averaged over $\mathcal{T}=500$ realizations of $\mathbf{D}$ and $\G_{\rm NLoS}$. Moreover, we assume the \gls{ue}s are uniformly distributed in a circular area with the center in $(100;\,50)\,m$ and a radius of $r=15$ $m$, while the \gls{ris} is located at $(100;\, 0)\,m$ and the \gls{bs} is located at $(0;\, 0)\,m$. In each setup, the \gls{aod} angles were uniformly distributed as: $\phi^{r,k}_{\rm{AoD}} \sim \mathcal{U}[0,\frac{\pi}{2}]$, $\varphi^{r,k}_{\rm{AoD}} \sim \mathcal{U}[-\frac{\pi}{3},\frac{\pi}{3}]$, and $\varphi^{bs}_{{\rm AoD}} \sim \mathcal{U}[-\frac{\pi}{2},\frac{\pi}{2}]$; besides, the \gls{aoa} angles $\phi^r_{{\rm AoA}} \sim \mathcal{U}[0,\frac{\pi}{2}]$ and $\varphi^r_{{\rm AoA}} \sim \mathcal{U}[-\frac{\pi}{2},\frac{\pi}{2}]$. All presented results for statistical \gls{csi}-optimization have been averaged over $\mathcal{S}=50$ different setups. Table \ref{t:param} summarizes the adopted values for the main simulation parameters.
 
\begin{table}[htbp!] 
\caption{Adopted Simulation Parameters.}
\small
\label{t:param} 
\begin{center}
\begin{tabular}{ll}
\toprule
\bf Parameter  & \bf  Value\\
\hline \hline
\multicolumn{2}{c}{\bf \gls{ris}-Aided \gls{m-mimo} System}\\
\hline
Max. Power Budget at \gls{bs}                 &   $P_{\textsc{tx}}  \in \{20;\, 50\}$ [dBm]\\
Noise variance                          & $\sigma^2 = -95$ dBm  \\
Numbers of \gls{ue}s                          & $K = 10$\\
Max. \# Antennas at BS                       &   $M_{\max} = 256$\\
\# Reflecting meta-surfaces             &   $N =  100$\\
Target rate                             &   $R_{\min} = 1$ bps/Hz\\
Power inefficiency                      &   $\varrho$ = 1.2 \cite{EEdebbah}\\
Fixed power consumption                 & $P_{{\rm FIX}}$ = 9 dBW \cite{EEdebbah}\\
Power \gls{bs} antenna activation       & $P_{\rm{BS}}$ = 1 W \cite{7031971}\\
Power \gls{ris} element activation       & $P_{\rm{RIS}}$ = 10 dBm \cite{EEdebbah,EEli}\\
\hline
\multicolumn{2}{c}{ \bf Channel Parameters}\\
\hline
Path-loss models& $\alpha_G = -25\log
(d_g)$\\
& $\alpha_{F,k} = -10.6-20\log(d_{F,k})$ \\
& $\alpha_{D,k}=-35.6 - 40\log(d_{D,k})$
\\
Matrix $\G$ (\gls{bs}-\gls{ris}) & Rician fading\\
Matrix $\D$ (\gls{bs}-\gls{ue}s) & Rayleigh fading\\
Matrix $\F$ (\gls{ris}-\gls{ue}s) & Purely \gls{los}\\
Rician coefficient  & $K_1=3.5$\\
\hline
\gls{mcs} & $\mathcal{T} = {500}$ realizations\\
Setups & $\mathcal{S}=50$ setups\\
\toprule
 \end{tabular}
 \end{center}
\end{table}

Aiming to highlight the proposed method, we simulate four different variable optimization strategies denoted as
\begin{enumerate}
\item {\bf p  optz}: only the power of the \gls{ue}s, $\mathbf{p}$, is optimized;
\item \textbf{p, v optz}.:  the power of the \gls{ue}s, $\mathbf{p}$, and the \gls{ris} phase shifts, $\mathbf{v}$, are optimized;
\item \textbf{p, {\bf \it M} optz}.: The power of the \gls{ue}s, $\mathbf{p}$, and the number of antennas, $M$, are optimized;
\item \textbf{p, v, {\bf \it M} optz}.: The power of the \gls{ue}s, $\mathbf{p}$, the number of antennas $M$, and the \gls{ris} phase shifts 
$\mathbf{v}$ are optimized sequentially.
\end{enumerate}

It should be noted that when $M$, $\mathbf{v}$ or $\mathbf{p}$ are not optimized, they are assigned random feasible values. In addition, for comparison purposes, we also consider the {\bf random all} allocation approach, where all variables ({$\mathbf{p}$, $\mathbf{v}$, and $M$}) are assigned random feasible values. Specifically, $\mathbf{p}$ is assigned values that satisfy constraint $\eqref{c:1}$, $M$ obeys the constraint $\eqref{c:3}$, and $\mathbf{v}$ adheres to constraint $\eqref{c:4}$. {In the forthcoming sub-sections, the numerical results will be presented as follows: we explore the achievable \gls{ee} of the \gls{ris}-aided \gls{m-mimo} system alongside its corresponding average number of active antennas, power consumption, and attained rates. This structured approach aims to provide a clear and organized understanding of the obtained outcomes we are interested in.}

\subsection{Attainable Energy Efficiency} \label{ss:avEE}

Fig. \ref{fig:avEE} illustrates the achievable \gls{ee} performance as a function of the transmit 
power budget at \gls{bs}, $\ptx$. We evaluated the performance for the four variable optimization strategies: 1) only $\mathbf{p}$; 2) $\mathbf{p}$ and $\mathbf{v}$; 3) $\mathbf{p}$ and $M$; and 4) $\mathbf{p}$, $\mathbf{v}$ and $M$ optimization, denoted as yellow, red, green, and blue color respectively. As benchmarks, we adopted the random all ($\mathbf{p}$, $\mathbf{v}$, and $M$) strategy denoted herein as black color curves.

\begin{figure}[!ht]
\normalsize
\centering
\includegraphics[width=8.85cm]{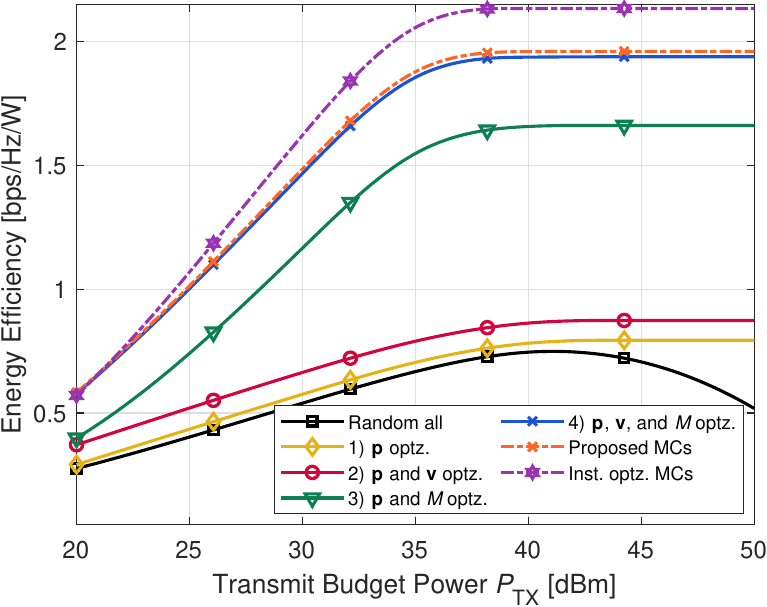}
\caption{Average \gls{ee} {\it vs} the transmit 
power budget at the \gls{bs} ($\ptx$). Performance evaluation of the proposed algorithm for four different approaches: 1) only $\mathbf{p}$; 2) $\mathbf{p}$ and $\mathbf{v}$; 3) $\mathbf{p}$ and $M$; and 4) $\mathbf{p}$, $\mathbf{v}$ and $M$ optimization denoted as yellow, red, green, and blue color respectively. Random $\mathbf{p}$, $\mathbf{v}$, and $M$ and the instantaneous optimization strategy are denoted as black and purple, respectively.}
\label{fig:avEE}
\end{figure}

{Moreover, we adopted the instantaneous \gls{csi}-based optimization strategy, for gap-performance comparison purposes, which has been implemented based on Algorithm 1 proposed in \cite{9966649} represented as purple color curves, which is projected specifically for the scenarios multi-user \gls{ris}-aided with \gls{zf} system, as the same adopted herein.} We also plotted the average instantaneous \gls{ee} obtained from the proposed statistical \gls{csi}-based optimization approach, evaluated from the \gls{mcs}, and depicted in orange color, in order to highlight the efficacy of the proposed method.

First and foremost, it must be noticed the consistency of the \gls{ee} computed utilizing the \gls{er} lower bound, Eq. \eqref{eq:rateLB}, which undoubtedly proves that $R_k > \widetilde{R}_k$. This is evident as the obtained \gls{ee} \gls{mcs} curve is slightly higher than the \gls{ee} evaluated from the \gls{er} equation, demonstrating the tightness of the derivation proposed in \cite{Zhi2022} and affirming the potential of the statistical \gls{csi}-based optimization for enhancing \gls{ee}. Regarding the different strategies for statistical \gls{csi}-based optimization, it is evident that the proposed algorithm for strategy 4) can outperform all other strategies with statistical \gls{csi} optimization in terms of \gls{ee} enhancement, achieving a significant improvement by increasing the \gls{ee} from $\approx 0.74$ bps/Hz/W with the random solution (in the maximum value) to $\approx 1.96$ bps/Hz/W. This translates to a remarkable \gls{ee} gain of approximately $165\%$. Since it optimizes all variables $\mathbf{p}$, $M$, and $\mathbf{v}$ sequentially, such a result was expected. Besides, given the high value assumed by $M_{\max}$ as in practical \gls{ris}-aided \gls{m-mimo} scenarios, there is a significant probability of obtaining a high value of $M$ in the random realizations for both strategies 1) and 2) as depicted in Fig. \ref{fig:nbrAntennas}, and analyzed in details in subsection \ref{sec:M_optz}. Thus, when $M$ is high, the energy consumption is significantly increased implying poor \gls{ee}. Accordingly, optimizing $M$ becomes of paramount importance. 

{A notable trend becomes evident when examining the \gls{ee} of all schemes, except the Random all strategy. Initially, \gls{ee} experiences an upward trajectory before stabilizing as the transmit budget power $P_{\rm TX}$ of \gls{bs} increases. This behavior arises because \gls{ee} does not inherently exhibit a strict monotonically increasing relationship with the transmit budget power. Therefore, when $P_{\rm TX} \geq 40$ dBm, the algorithm can control any surplus transmit power once it is superfluous as its utilization would compromise \gls{ee}. This explains why the Random all scheme provides a pronounced decline in \gls{ee} performance as the transmit budget power scales significantly since this approach deploys all the power available} \footnote{Power allocation is not subject to optimization in Random all approach.}.

Moreover, in Fig. \ref{fig:avEE}, one can infer that the impact of optimizing the \gls{ris} phase shifts is not negligible on the attainable \gls{ee}, even considering the statistical optimization. It is noteworthy that the achieved improvement in scenarios in which the number of active \gls{bs} antennas $M$ assumes low values (optimized $M$) is considerably greater compared to those scenarios where $M$ assumes high values. This can be readily verified by examining the obtained gain in the \gls{ee} curve of strategy 1) compared to the curve of strategy 2), as well as in the curve of strategy 3) concerning the curve of strategy 4). 

Based on Fig. \ref{fig:avEE}, one can see the potential to achieve higher \gls{ee} through the statistical \gls{csi} optimization approach since this methodology enables us to reduce the overhead associated with the phase shifts optimization during each channel coherence time while maintaining the fixed $\mathbf{v}$ over multiples channel coherence time with a minimal performance loss. As a result, the \gls{ee} decreases from 2.13 to 1.96 bps/Hz/W, representing a loss of $\approx 7.98\%$, in high power regimes.

\subsection{Average Number of Active Antennas}\label{sec:M_optz}

Fig. \ref{fig:nbrAntennas} shows the number of average active antennas $M$ for the optimization strategies 1), 2), 3), and 4), {as well as for the Random all scheme}. It is noteworthy that strategies {Random all}, 1), and 2) exhibit a significantly higher average number of active antennas compared to the average number of active antennas in strategies 3) and 4). 
It occurs due to the random assignment for $M$ in each setup $\mathcal{S}$, and since strategies {Random all}, 1), and 2) do not optimize $M$, there is a high probability for $M$  to be higher than that 
in strategies 3) and 4). This finding helps explain why strategies {Random all}, 1), and 2) attain poor \gls{ee} as illustrated in Fig. \ref{fig:avEE}, and higher rates as analyzed further ahead in Fig. \ref{fig:rateCCDF}.

\begin{figure}[ht!] 
\centering
\normalsize
\includegraphics[width=8.85cm]{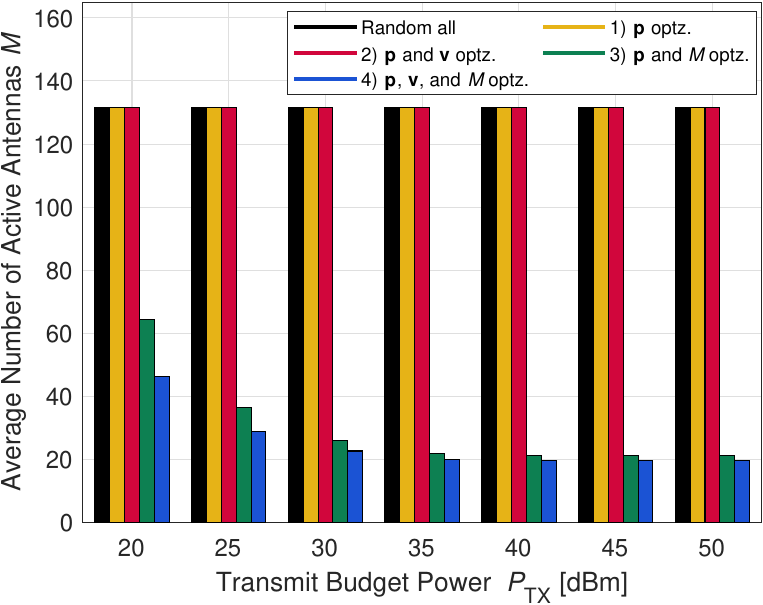}
\caption{{Average number of antennas for four optimization strategies.}}
\label{fig:nbrAntennas}
\end{figure}

Furthermore, by optimizing the \gls{ris} phase shifts 
remarkably reduce the average number of active BS antennas, mainly in low-power regimes, $e.g.$, by comparing strategy 3) with strategy 4), at $\ptx=20$ dBm, the average number of antennas is reduced from 131 to 46. These results highlight the potential of \gls{ris} deployment and optimization.

\subsection{Power Consumption} \label{ss:avPower}

Fig. \ref{fig:avPower} depicts the average percentage of the transmit power budget utilized for transmission versus the transmit power budget. In the low-budget power regime, all strategies deplete the available power budget; however, as the transmit
power budget increases beyond 40 dBm, the algorithm starts to minimize the power consumption. This behavior is expected since aiming to achieve the maximum \gls{ee}, the available power budget should not be fully utilized. 

\begin{figure}[!ht]
\centering
\normalsize
\includegraphics[width=8.5cm]{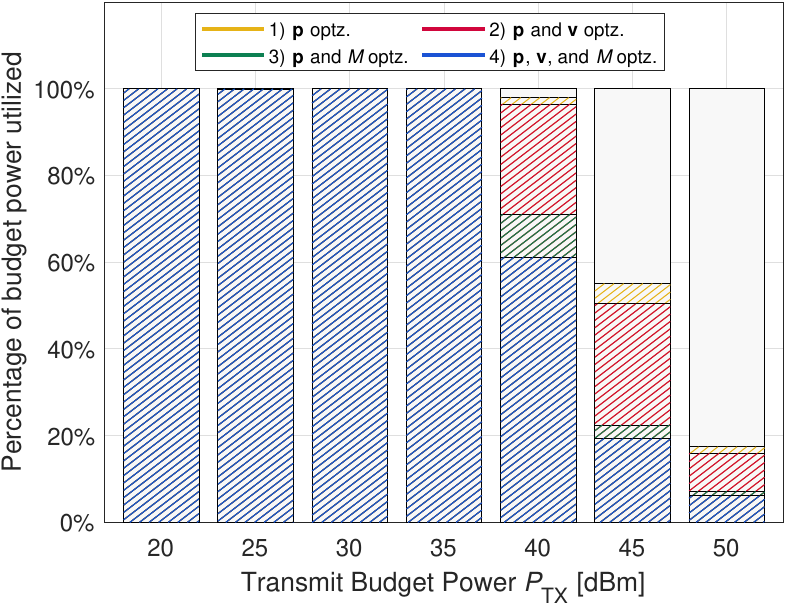}
\caption{Percentage of the transmit power budget utilized {\it vs} the transmit power budget at \gls{bs}. The colors in the graph are associated with the strategy according to Fig. \ref{fig:avEE}. }
\label{fig:avPower}
\end{figure}

Additionally, this strategy is directly accountable for the constant \gls{ee} observed in high-power budget regimes, as presented in Fig. \ref{fig:avEE}. Furthermore, it becomes apparent that optimizing the number of antennas $M$ is of utmost importance to save a significant amount of power, reducing the power consumption substantially, from $100\%$ (as in the Random all Strategy case) to $61\%$, $19\%$, and $6\%$ in the best strategy, for $\ptx=$ 40, 45, and 50 dBm respectively. Consequently, this strategy also implies high \gls{ee} gain, as presented in Fig. \ref{fig:avEE}. Moreover, one can see the substantial impact on power consumption reduction caused by optimizing the \gls{ris} phase shifts. This power reduction 
corroborates the potential of \gls{ris} deployment and configuration, resulting in a remarkable \gls{ee} gain that cannot be obtained in non-\gls{ris} scenarios.

\subsection{Distribution on the Attainable Data Rates}\label{ss:rateCCDF}

\begin{figure*}[!htb]
\centering
\begin{minipage}[b]{0.48\textwidth}
\centering
\includegraphics[width=8.3cm]{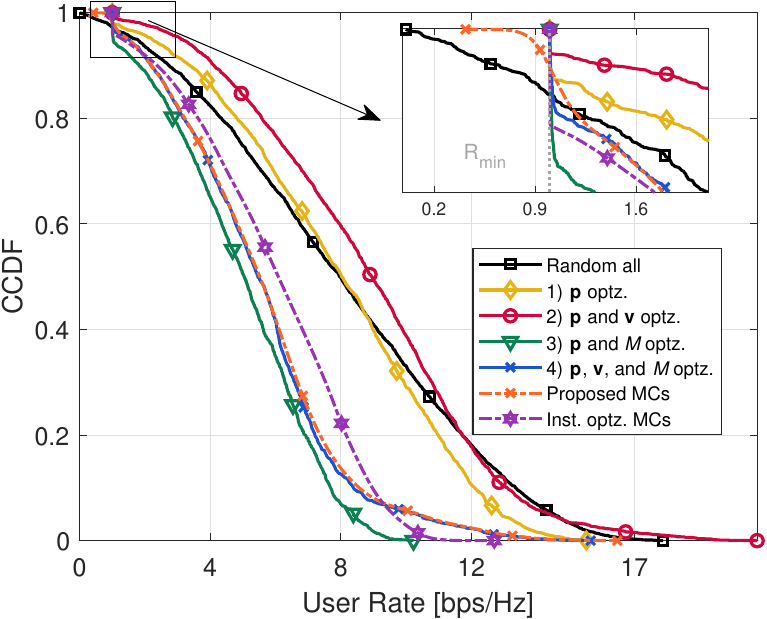}
\caption{\gls{ccdf} of the \gls{er} $\widetilde{R}_k$ for four different strategies with statistical \gls{csi}-based optimization (solid lines) and the instantaneous \gls{ue} rate $R_k$ for instantaneous \gls{csi}-based optimization (dashed lines), with $R_{\min}$ = 1 bps/Hz.}
\label{fig:rateCCDF}
    \end{minipage}%
    \quad
    \begin{minipage}[b]{0.48\textwidth}
        \centering
    \includegraphics[width=8.3cm]{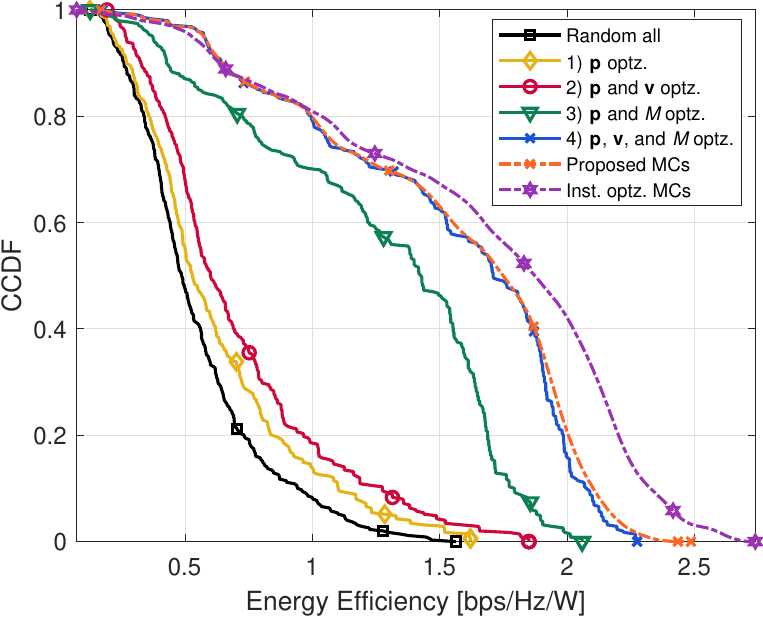}
        \caption{{\gls{ccdf} of the \gls{ee} for four different optimization strategies with statistical \gls{csi}-based optimization (solid lines) and for instantaneous \gls{csi}-based optimization (dashed lines), with $R_{\min}=1$ bps/Hz.}}
        \label{fig:EECCDF}
    \end{minipage}
\end{figure*}

Fig. \ref{fig:rateCCDF} reveals the \gls{ccdf} of the \gls{ue}' \gls{er} $\widetilde{R}_k$ over $\mathcal{S}=50$ different realizations setups, represented by the solid lines. Besides,
dashed lines indicate the computed instantaneous \gls{ue}' rate, $R_k$, over $\mathcal{T}=500$ \gls{mcs} for $\mathbf{D}$ and $\mathbf{G}_{{\rm NLoS}}$, in the range of $\ptx$ from 20 to 50 dBm. Notice that the proposed Statistical \gls{csi}-based optimization method necessarily obeys the constraint \eqref{c:2} for all Strategies solutions. This fulfillment is primordial to 
guarantee appropriate \gls{qos} for all \gls{ue}s. Moreover, one can see that Strategies {Random all}, 1), and 2) are capable of providing higher rates than Strategies 3) and 4). 

This is due to the consumed power in Strategies 1) and 2) being higher than that in Strategies 3) and 4), as shown in Fig. \ref{fig:avPower}. However, one can see that although Strategy 4) has a much lower number of active antennas and lower power consumption than Strategy 3), Strategy 4) still has a feasible probability of attaining higher rates than Strategy 3). This is due to the \gls{ris} phase shifts optimization, which can enhance some \gls{ue}s by maximizing the \gls{er}, as discussed in the subsection \ref{ss:Phase}. We can also notice from Fig. \ref{fig:rateCCDF} that the instantaneous \gls{ue} rate $R_k$ obtained through statistical \gls{csi}-based optimization (orange dashed curve) can achieve lower rates than $R_{\min}$, implying a violation of \eqref{c:2}, in contrast to the data rates obtained through the instantaneous \gls{csi}-based optimization (purple dashed curve), which are always higher or equal than $R_{\min}$.  

{Additionally, in Fig. \ref{fig:EECCDF}, we present the \gls{ccdf} of \gls{ee}. This graphical representation corroborates that optimization strategies 3) and 4) can effectively fine-tune the number of active antennas and power consumption through the proposed algorithm, leading to enhanced \gls{ee} at the cost of a tolerable reduction in data rates. Furthermore, it is evident that strategies such as Random all, 1), and 2), where the number of antennas is not optimized, yield lower \gls{ee} rates. This is attributed to the fact that the majority of setups in $\mathcal{S}$ feature a higher number of active BS antennas in comparison to the active antennas achieved by strategies 3) and 4).}

\subsection{Rician Coefficient Factor} \label{sec:RicianFactor}

{Fig. \ref{fig:Rician} illustrates the achievable \gls{ee} performance as a function of power budget $\ptx$ for three different Rician factors $K_1$. Initially, it becomes evident that in the worst-case scenario, namely when $K_1 = 0$, the greatest disparity between statistical \gls{csi} optimization and instantaneous \gls{csi} optimization occurs. 

\begin{figure}[ht!] 
\centering
\normalsize
\includegraphics[width=8.7cm]{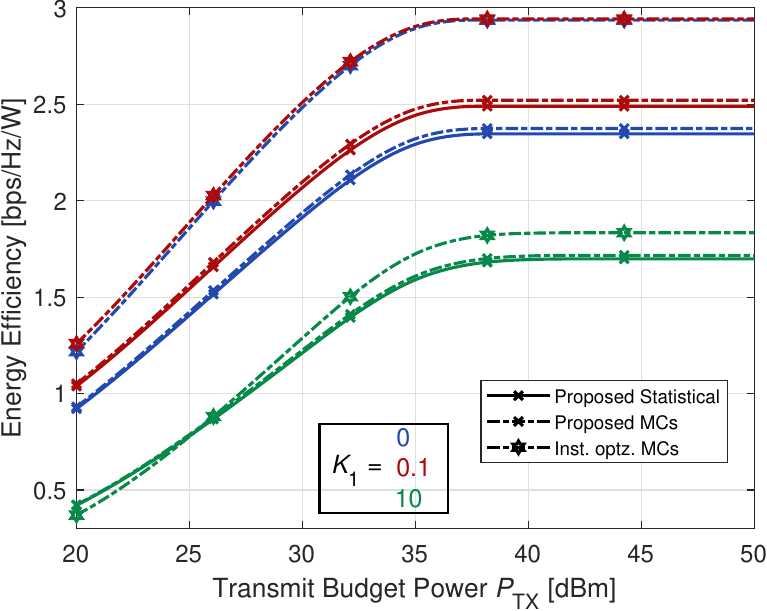}
\vspace{-.4cm}
\caption{Average \gls{ee} {\it vs} the transmit power budget at the BS ($\ptx$). Performance evaluation of the proposed algorithm for three different values of Rician factor $K_1$.}
\label{fig:Rician}
\end{figure}

This observation is justified by considering a scenario with no \gls{los} component for the \gls{bs}-\gls{ris} link. In such a case, optimizing the \gls{ris} phase shifts becomes inefficient, where the constants in \eqref{eq:B} and \eqref{eq:Ak} reduce, respectively, to $\mathbf{B} = \frac{1}{N} \mathbf{I}_N$ and $\mathbf{A}_k \approx \frac{1}{N(N \alpha_G \alpha_{F,k}+ \alpha_{D,k})} \mathbf{I}_N$. Therefore, the statistical \gls{csi} optimization strategy becomes irrelevant in this scenario. Furthermore, a slight increase in the Rician factor can enhance the system's performance, underscoring the significant impact of optimizing the \gls{ris} phase shifts. However, it becomes apparent that the system's performance deteriorates when the Rician factor reaches considerably high values. This can be attributed to the dominance of the \gls{los} component in the system, wherein the channel matrix $\mathbf{G}$ becomes rank-deficient, leading to a reduction in multi-path diversity and consequently worsening the overall system performance.}

\subsection{{Computational} Complexity of Algorithm \ref{alg:phase-optz}}
In Fig. \ref{fig:nbrIterations}, we illustrate the \gls{er} {(objective function given by Eq. \eqref{eq:rateLB}} {\it vs.} the average number of iterations ($I_1$) attainable with the proposed procedure for \gls{ris} phase shifts optimization, Algorithm \ref{alg:phase-optz}. Additionally, we include a benchmark utilizing the Gradient approach algorithm proposed in \cite{Zhi2022}. {Herein, a backtrack line search strategy for step size determination is adopted, as in \cite{Zhi2022}.} {To ensure fairness in our comparison, we apply a consistent stopping criterion for both the proposed algorithm and the gradient-based approach. Specifically, the algorithm will terminate when the absolute difference between the objective function's values at iterations $n+1$ and $n$, denoted as $f^{(n+1)}$ and $f^{(n)}$ respectively, becomes less than $10^{-3}$, $i.e.$, $|f^{(n+1)}-f^{(n)}|<10^{-3}$.} It is apparent that the proposed algorithm achieves a gradual convergence with significantly fewer iterations than the Gradient approach, {\it i.e.,} the proposed optimization method for both \gls{sfp} and Analytical strategies needs $I_1 \approx {20}$ iterations for convergence. Moreover, one can observe that the proposed algorithm outperforms the Gradient approach by achieving $\sum_k\widetilde{R}_k\approx {42}$ bps/Hz, while the latter achieves around {$40$} bps/Hz after {80} iterations.

\begin{figure}[ht!] 
\centering
\normalsize
\includegraphics[width=8.85cm]{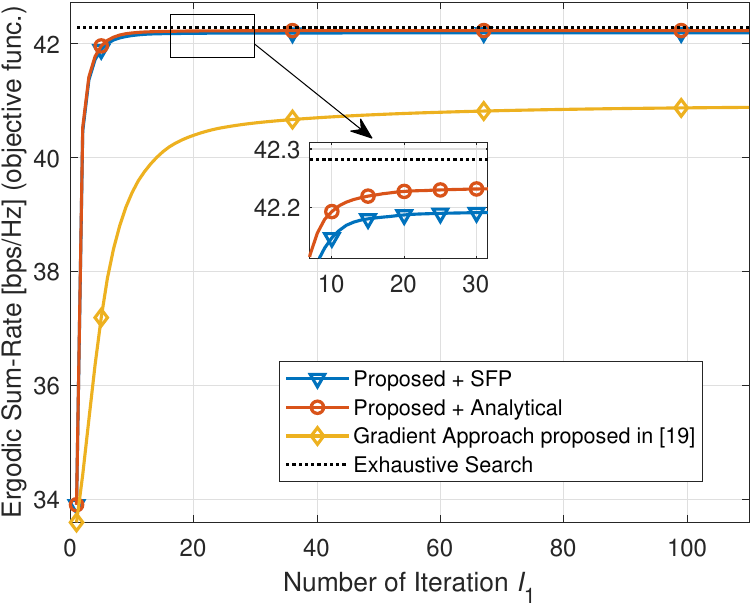}
\caption{The convergence behavior of the proposed Algorithm \ref{alg:phase-optz} for \gls{ee} maximization and the convergence behavior of the Gradient approach benchmark proposed in \cite{Zhi2022} over $\mathcal{S}=1000$ different setups. Here we adopt $M=M_{\max}$, $\ptx=30$ dBm, {$K=10$} and {$N=25$}.}
\label{fig:nbrIterations}
\end{figure}

We also can notice in Fig. \ref{fig:nbrIterations} the potential of the proposed scheme 
with the \gls{sfp} method solution since the performance gap is negligible regarding the optimal solution given analytically by Eq. \eqref{eq:thetaopt}, but resulting in a
significantly lower complexity according to Table \ref{tab:complexity} and \ref{t:time}. This clearly highlights the potential of the proposed solution. {Furthermore, Table \ref{t:time} confirms that the proposed solution with \gls{sfp} is appealing since the running time is lower than both the Gradient approach and instantaneous \gls{csi}-optimization}\footnote{{For instantaneous \gls{csi} optimization, we computed the average time across all $\mathcal{T}=500$ Monte Carlo simulations and $\mathcal{S}=50$ different setups.}}.
 
\begin{table}[!ht] 
\caption{Running Time Comparison of the Algorithms in Seconds}
\centering
\begin{tabular}{l|c|} 
\cline{2-2}
& \bf Elapsed Time [s]\\ \hline
\multicolumn{1}{|l|}{Proposed + \gls{sfp}}            & 0.0287             \\ \hline
\multicolumn{1}{|l|}{Proposed + Analytical}     & 0.0572               \\ \hline
\multicolumn{1}{|l|}{Gradient Approach \cite{Zhi2022}} & 0.0688               \\ \hline
\multicolumn{1}{|l|}{Instantaneous \gls{csi} Optz.} & 0.0732             \\ \hline
\end{tabular} \label{t:time}
\end{table}

\section{Conclusions}\label{sec:concl}

This work addresses the \gls{ee} maximization problem in a downlink \gls{ris}-aided \gls{m-mimo} communication system operating under generalized Rician \gls{bs}-\gls{ris} channels and \gls{zf} precoding. We proposed a statistical \gls{csi}-based optimization strategy that treats the \gls{ris} phase shifts and power allocation as a constant over multiple channel coherence time intervals, reducing the overhead of the variables optimization in every channel coherence time. Based on an \gls{er} lower-bound closed-form expression, the \gls{ee} optimization problem with power and \gls{qos} constraints is formulated. To circumvent the challenging coupled optimization variables and non-convex constraints, the proposed alternating-sequential optimization approach, based on the fractional programming techniques and dual Lagrangian method, optimizes the power allocation vector ($\bf p$), the number of \gls{bs} antennas ($M$), and the \gls{ris} phase shifts ($\bf v$) in an iterative manner based on analytical closed-form expressions until convergence. Our approach
offers an elegant methodology to gradually improve the system \gls{ee} performance by optimizing the system variables, demonstrating the potential of \gls{ris} deployment and optimization in \gls{m-mimo} systems since the proposed scheme can result in high data rates for the \gls{ue}s, lower power consumption and lower active antennas operating at \gls{bs}. Our numerical results validate the effectiveness of the proposed Algorithm, as it substantially reduces the computational complexity while maintaining an exciting trade-off in terms of \gls{ee} performance. Our approach incurs a performance loss of $\approx 7.98\%$, compared to the instantaneous \gls{csi}-based optimization method. Besides, our proposed Algorithm can achieve better \gls{ee} gains with lower complexity. 

\appendices

\section{Proof of Lemma \ref{lemma:1} } \label{app:lemma} 
Initially, we should notice the $\mathcal{P}_3$ can be equivalently transformed into
\begin{equation} \max_{\mathbf{v}} \;  \mathbf{v}^H \mathbf{C} \mathbf{v} = \max_{\mathbf{v}} \; {\rm Tr} \left(\mathbf{C} \mathbf{V} \right),
\end{equation}
with $\mathbf{V}=\mathbf{v}\mathbf{v}^H$, where the following constraints ${\rm rank}(\mathbf{V})=1$ $[\mathbf{V}]_{n,n}=1$, $\forall n\in \{1,\dots,N\}$, must be attained.
Proceeding, conveniently, we can rewritten $\mathbf{v}$ as the following manner
\begin{equation} \label{eq:appAeq1}
    \mathbf{v} = [\mathbf{v}_1^H,  e^{-j\theta_n},  \mathbf{v}_2^H]^H.
\end{equation}
where $\mathbf{v}_1$ $\in \mathbb{C}^{n-1}$ and $\mathbf{v}_2$ $\in \mathbb{C}^{N-n}$ are respectively defined as $\mathbf{v}_1 \triangleq \left[\mathbf{v}\right]_{1:(n-1)}$ and $\mathbf{v}_2 \triangleq \left[\mathbf{v}\right]_{(n+1):N}$. 
This alternative notation enables us to look for the $n$-th element of $\mathbf{v}$, $e^{j\theta_n}$. Therefore, by utilizing Eq. \eqref{eq:appAeq1}, one can rewrite $\mathbf{V} $ as:

\begin{equation}
    \mathbf{V} = 
 \mathbf{v}\mathbf{v}^H = \begin{bmatrix}
\mathbf{v}_1 \mathbf{v}_1^H & \mathbf{v}_1 e^{-j\theta_n} & \mathbf{v}_1 \mathbf{v}_2^H \\ 
e^{j\theta_n} \mathbf{v}_1^H & 1 & e^{j\theta_n} \mathbf{v}_2^H \\ 
\mathbf{v}_2 \mathbf{v}_1^H & \mathbf{v}_2 e^{-j\theta_n} & \mathbf{v}_2 \mathbf{v}_2^H
\end{bmatrix}.
\end{equation}

By defining the following variables 
\begin{equation} \label{eq:appAeq3}
\left \{
\begin{aligned}
\mathbf{C}_1 &= \left[\mathbf{C}\right]_{1:(n-1),1:(n-1)}  &\quad & \in \mathbb{C}^{(n-1) \times (n-1)}\\ 
\mathbf{c}_2 &= \left[\mathbf{C}\right]_{1:(n-1),n} &\quad  & \in \mathbb{C}^{(n-1) \times 1}\\ 
\mathbf{C}_3 &= \left[\mathbf{C}\right]_{1:(n-1),(n+1):N} &\quad & \in \mathbb{C}^{(n-1) \times (N-n)}\\ 
C_4 &=  \left[\mathbf{C}\right]_{n,n} &\qquad \quad  & \in \mathbb{C} \\ 
\mathbf{c}_5 &= \left[\mathbf{C}\right]_{n,(n+1):N} &\quad   &\in \mathbb{C}^{1 \times (N-n)} \\
\mathbf{C}_6 &= 
\left[\mathbf{C}\right]_{(n+1):N,(n+1):N} &\quad &\in \mathbb{C}^{(N-n) \times (N-n)},
\end{aligned}
\right.
\end{equation}
the matrix $\mathbf{C}$ can be rewritten in terms of these variables:

\begin{IEEEeqnarray}{rrCl}
\mathbf{C} = \begin{bmatrix}
\mathbf{C}_1 & \mathbf{c}_2 & \mathbf{C}_3 \\ 
\mathbf{c}_2^H &  
C_4 & 
\mathbf{c}_5 \\ 
\mathbf{C}_3^H & \mathbf{c}_5^H & \mathbf{C}_6 
\end{bmatrix},
\end{IEEEeqnarray}

Therefore, after some basic algebraic manipulations: 
\begin{IEEEeqnarray}{rCl} \label{eq:appAeq2}
     {\rm Tr}(\mathbf{C} \mathbf{V}) &=&  e^{-j\theta_n} ({\rm Tr}(\mathbf{c}_2^H \mathbf{v}_1) + {\rm Tr}(\mathbf{c}_5 \mathbf{v}_2 ))    
    \\
    &+& e^{j\theta_n}({\rm Tr} (\mathbf{c}_2 \mathbf{v}_1^H) + {\rm Tr}(\mathbf{c}_5^H \mathbf{v}_2^H)) + C_4 \nonumber
    \\
     && \hspace{-1.8cm} + {\rm Tr}\left( \mathbf{C}_3 \mathbf{v}_2 \mathbf{v}_1^H
    + \mathbf{C}_1 \mathbf{v}_1 \mathbf{v}_1^H 
    + \mathbf{C}_3^H \mathbf{v}_1 \mathbf{v}_2^H   
 \right) + {\rm Tr}\left( \mathbf{C}_6 \mathbf{v}_2\mathbf{v}_2^H \right), \nonumber 
\end{IEEEeqnarray}

Interestingly, Eq. \eqref{eq:appAeq2} leads us to see that maximize $\mathcal{P}_3$ w.r.t. $\theta_n$ is equivalently to
\begin{IEEEeqnarray}{rCl}
    \max_{\theta_n} \; {\rm Tr}(\mathbf{C} \mathbf{V}) =  \max_{\theta_n} \; 2 \mathbb{R}\{ e^{-j\theta_n}({\rm Tr}(\mathbf{c}_2^H \mathbf{v}_1) + {\rm Tr}(\mathbf{v}_2 \mathbf{c}_5)) \} \nonumber
\end{IEEEeqnarray}
whose optimal solution is straightforwardly given in closed-form by $\theta_n^\star = \angle \left({\rm Tr}(\mathbf{c}_2^H \mathbf{v}_1) + {\rm Tr}(\mathbf{v}_2 \mathbf{c}_5)
\right)$. Substituting the values of $\mathbf{c}_2$ and $\mathbf{c}_5$ given in Eq. \eqref{eq:appAeq3}, and $\mathbf{v}_1$ and $\mathbf{v}_2$, the proof is complete.

\end{document}

%% file: acronyms.tex
\newacronym{3d}{3-D}{three-dimensional}
\newacronym{3gpp}{3GPP}{3rd Generation Partnership Project}
\newacronym{5g}{5G}{fifth generation}
\newacronym{ao}{AO}{alternating optimization}
\newacronym{aoa}{AoA}{angle of arrival}
\newacronym{aod}{AoD}{angle of departure}
\newacronym{as}{AS}{antenna selection}
\newacronym{awgn}{AWGN}{additive white gaussian noise}
\newacronym{bs}{BS}{base station}
\newacronym{csi}{CSI}{channel state information}
\newacronym{chest}{CHEST}{channel estimation}
\newacronym{ccdf}{CCDF}{complementary cumulative distribution function}
\newacronym{dl}{DL}{downlink}
\newacronym{ee}{EE}{energy efficiency}
\newacronym{er}{ER}{ergodic rate}
\newacronym{fp}{FP}{fractional programming}
\newacronym{iid}{i.i.d.}{independent and identically distributed}
\newacronym{kkt}{KKT}{Karush–Kuhn–Tucker}
\newacronym{ldt}{LDT}{Lagrangian Dual Transform}
\newacronym{los}{LoS}{line-of-sight}
\newacronym{mimo}{MIMO}{multiple-input multiple-output}
\newacronym{m-mimo}{mMIMO}{massive multiple-input multiple-output}
\newacronym{mrc}{MRC}{maximum ratio combining}
\newacronym{mm}{MM}{majoration-minimization}
\newacronym{mcs}{MCs}{Monte-Carlo simulation}
\newacronym{nlos}{NLoS}{non-line-of-sight}
\newacronym{nlp}{NLP}{non-linear problem}
\newacronym{pa}{PA}{power amplifier}
\newacronym{pc}{PC}{pilot contamination}
\newacronym{qos}{QoS}{quality of service}
\newacronym{rf}{RF}{radio frequency}
\newacronym{re}{RE}{resource efficiency}
\newacronym[plural=RISs]{ris}{RIS}{reconfigurable intelligent surface}
\newacronym{se}{SE}{spectral efficiency}
\newacronym{sfp}{SFP}{sequential fractional programming}
\newacronym{sinr}{SINR}{signal-to-interference-plus-noise ratio}
\newacronym{snr}{SNR}{signal-to-noise ratio}
\newacronym{sdr}{SDR}{semidefinite relaxation}
\newacronym{tdd}{TDD}{time-division duplex}
\newacronym[plural=UEs, firstplural=users' equipment (UEs)]{ue}{UE}{user's equipment}
\newacronym{ul}{UL}{uplink}
\newacronym{upa}{UPA}{uniform planar array}
\newacronym{uspa}{USPA}{uniform squared planar array}
\newacronym{ula}{ULA}{uniform linear array}
\newacronym{vr}{VR}{visibility region}
\newacronym{wsr}{WSR}{weighted sum-rate}
\newacronym{xl-mimo}{XL-MIMO}{extra-large scale massive MIMO}
\newacronym{zf}{ZF}{zero-forcing}